\documentclass[journal]{IEEEtran}
\IEEEoverridecommandlockouts
\usepackage{amsmath,amsfonts}
\usepackage{algorithmic}
\usepackage{algorithm}
\usepackage{array}
\usepackage{textcomp}
\usepackage{stfloats}
\usepackage{url}
\usepackage{verbatim}
\usepackage{graphicx}
\usepackage{cite}
\usepackage{float}
\usepackage[table]{xcolor}
\newcommand{\ra}[1]{\renewcommand{\arraystretch}{#1}}
\usepackage{booktabs}
\usepackage{caption}
\usepackage{epsfig}
\usepackage{epstopdf}
\usepackage{subfigure}
\usepackage{amssymb}
\newtheorem{theorem}{Theorem}
\newtheorem{lemma}{Lemma}

\allowdisplaybreaks

\begin{document}

\title{{J}oint {L}ocalization and {O}rientation with {T}riple-{B}eam {F}ingerprints in {M}assive {MIMO-OFDM}}

\author{Yu~Zhao,~\IEEEmembership{Graduate Student~Member,~IEEE,}
	Zhenzhou~Jin,~\IEEEmembership{Graduate Student~Member,~IEEE,}
	Jinke~Tang,~\IEEEmembership{Graduate Student~Member,~IEEE,}
	Li~You,~\IEEEmembership{Senior~Member,~IEEE,}
	Chen~Sun,~\IEEEmembership{Member,~IEEE,}
	Xiang-Gen~Xia,~\IEEEmembership{Fellow,~IEEE,}
	and Xiqi~Gao,~\IEEEmembership{Fellow,~IEEE}

\thanks{Part of this work was presented at the 2026 IEEE Wireless Communications and Networking Conference Workshops\cite{YZLY2026}.
	
	Yu Zhao, Zhenzhou Jin, Jinke Tang, Li You, Chen Sun, and Xiqi Gao are with the National Mobile Communications Research Laboratory, Southeast University, Nanjing 210096, China, and also with the Purple Mountain Laboratories, Nanjing 211100, China (e-mail: yu\_zhao@seu.edu.cn, zzjin@seu.edu.cn, jktang@seu.edu.cn, lyou@seu.edu.cn, sunchen@seu.edu.cn, xqgao@seu.edu.cn).
	
	Xiang-Gen Xia is with the Department of Electrical and Computer Engineering, University of Delaware, Newark, DE 19716, USA (e-mail: xianggen@udel.edu).}
}

\maketitle

\begin{abstract}
With the widespread application of location-based services, fingerprint-based localization has demonstrated advantages in environments with complex signal propagation. Deep learning has significantly improved the efficiency of both offline training and online matching in localization processes. However, existing fingerprints only contain terminal position information without capturing motion states, and neural network designs have not fully incorporated structural features such as fingerprint sparsity. In this paper, we propose a triple-beam fingerprint (TBF) incorporating Doppler information and design a Transformer-based localization and orientation awareness network (LOA-Net) to simultaneously estimate user position and motion direction in massive multiple-input multiple-output (MIMO) orthogonal frequency division multiplexing (OFDM) systems. We first show the correlation between TBF and multipath information, and investigate the collinearity of different TBFs, demonstrating that TBF is an effective small-size sparse fingerprint. Then, we propose LOA-Net containing a mask-augmented detection Transformer for regression (MaskDETR-Reg) module and a fusion-enhanced Transformer for direction classification (Fusion-TDC) module to process angle-delay domain information and Doppler domain information, respectively. Finally, in the simulation of indoor scenarios defined in 3GPP 38.901, the proposed method achieves significantly better localization accuracy than weighted $K$-nearest neighbors (WKNN), 2D and 3D convolutional neural networks (CNNs), and achieves satisfactory motion direction estimation accuracy.
\end{abstract}


\section{Introduction}

With the development of information technology, location-based services have found increasingly extensive applications in the industrial Internet of Things and autonomous driving \cite{LYXQ2024,FZAG2019,YZLY2025,JSAG2025,ZJLY2026,ZJLY20253}. However, the performance of Global Positioning Systems (GPS), which has been widely adopted, significantly deteriorates in obstacle-dense urban or indoor environments. On the other hand, wireless networks inherently possess advantages in location and sensing as well, in addition to their primary function for communications \cite{CADB2021,YYLY2025,FGDD2021,YYLY2024}.

Wireless localization systems estimate positions based on reference signals between base station (BS) and user terminal (UT), which mainly fall into two categories: geometry-based and fingerprint-based methods. The geometry-based methods utilize triangular geometric properties to estimate locations through received signal strength (RSS) \cite{JTMR2015,MZYL2021,SFTL2008}, time of arrival (TOA) \cite{HWLL2022,KTXW2015,CXJH2017}, time difference of arrival (TDOA) \cite{SLMH2015} and direction of arrival (DOA) \cite{LQAL2024,YWKH2018,YZMS2019}, among others. The fingerprint-based methods uniquely map certain features of reference signals (i.e., fingerprints) to UT coordinates, determining the UT position by matching measured fingerprints against a fingerprint database. In complicated non-line-of-sight (NLOS) environments, multipath information severely degrades the performance of geometry-based methods but provides richer fingerprints. RSS is a commonly used fingerprint, but it only reflects limited channel information with low discriminability and instability \cite{XGXY2022,MLSS2024,XGAL2024}. Channel state information (CSI) contains multipath information in both angle and delay domains \cite{ZJLY20252,LYXG2016,JTLY2025,TCWW2021,LYKL2020,ZJLY2025,CXLY2025}, enabling the generation of more distinctive fingerprints to improve localization accuracy. As a key technology in 5G and next-generation communication systems, massive multiple-input multiple-output (MIMO) orthogonal frequency division multiplexing (OFDM) can capture high-resolution angle-delay domain multipath information for high-precision localization. Recently, increasing research attention has been paid to multi-domain channel characteristics. Among them, the triple-beam (TB) domain channel incorporates Doppler domain \cite{DSWW2021,JTXG2024,DSLS2023,XCLY2026}, which not only enables the generation of higher-resolution tensor fingerprints but also provides potential for UT motion-state awareness.

In addition to the inherent characteristics of fingerprints, the storage of fingerprint databases and matching methods are also critical factors affecting the performance of wireless localization. If the fingerprint matrix (tensor) is regarded as a single- (or multi-) channel image, and the coordinates are regarded as the labels, then the matching problem between them can be treated as a computer vision task. Consequently, many studies perform fingerprint localization through deep learning \cite{XSXG2018,XSCW2019,CWXY2021,BTMN2023}. Early research employed simple clustering algorithms such as the weighted $K$-nearest neighbors (WKNN) method to estimate positions \cite{XSXG2018}. Recent studies pay more attention to fingerprint features, adjust the neural network according to the fingerprint structure, and design more efficient localization methods, such as the deep convolutional neural network (CNN) \cite{XSCW2019}, the 3D CNN \cite{CWXY2021}, and the cooperative network based on the auto-encoder (AE) \cite{BTMN2023}.

However, the aforementioned methods fail to account for the impact of fingerprint sparsity on the network. Traditional CNNs rely on the local correlation of data, whereas sparse fingerprints disrupt spatial continuity, leading to convergence to local optima. Additionally, while previous studies have relied on inertial measurement units (IMU) \cite{BZHS2023,JCBZ2022} or continuous CSI recordings \cite{ZZML2021} to estimate UT motion states, these methods often require additional hardware or complex processing. In contrast, the TB domain channel exhibits inherent potential for direct and hardware-free motion-state perception. Therefore, we propose a method that leverages triple-beam fingerprint (TBF) to simultaneously achieve localization and orientation awareness in massive MIMO-OFDM systems. Table \ref{tb:InnovationComparison} presents a comparison between the proposed method and the existing methods.
\newcolumntype{L}{>{\hspace*{-\tabcolsep}}l}
\newcolumntype{R}{c<{\hspace*{-\tabcolsep}}}
\newcolumntype{C}{>{\centering\arraybackslash}p{0.8cm}}
\definecolor{lightblue}{rgb}{0.93,0.95,1.0}
\definecolor{lightgreen}{rgb}{0.95,1.0,0.93}
\begin{table}[htbp]
	\captionsetup{font=footnotesize}
	\caption{Fingerprint Matching Algorithm Comparison}\label{tb:InnovationComparison}
	\centering
	\ra{1.5}
	\scriptsize
	\begin{tabular}{LLCR}
		\toprule
		Technology & Orientation & Hardware & Limitations \\
		\rowcolor{lightblue}
		\midrule
		WKNN & \multicolumn{1}{c}{No} & No & Limited accuracy\\
		CNN & \multicolumn{1}{c}{No} & No & Poor handling of sparsity, local optima \\
		\rowcolor{lightblue}
		3D CNN & \multicolumn{1}{c}{No} & No & Poor handling of sparsity, local optima\\
		AE & \multicolumn{1}{c}{No} & No & High complexity, not tailored for sparsity \\
		\rowcolor{lightblue}
		IMU & \multicolumn{1}{c}{Yes} & Yes & Extra cost and power \\
		Continuous CSI & \multicolumn{1}{c}{Yes} & No & Error propagation\\
		\rowcolor{lightblue}
		\textbf{LOA-Net} (ours) & \multicolumn{1}{c}{Yes} & No & -- \\
		\bottomrule
	\end{tabular}
\end{table}

In this paper, we employ the TB domain channel power tensor as fingerprints, which contain information from the angle, delay, and Doppler domains. We first transform the space-frequency-time fingerprint (SFTF) into the TB domain to reduce fingerprint dimensions, and prove that TBFs at different locations maintain excellent discriminability. Through the concept of collinearity, we prove that TBFs with a smaller size achieve equivalent localization performance to SFTFs. Finally, we propose a Transformer-based localization and orientation awareness network (LOA-Net), which achieves superior positioning accuracy compared to conventional WKNN and CNN networks while maintaining effective motion direction estimation.

The main contributions are summarized as follows:

\begin{itemize}
	\item For massive MIMO-OFDM systems, we prove that TBF is related to multipath signal power, delay, angle of arrival, and Doppler frequency at corresponding locations, demonstrating good discriminability in scenarios with rich multipath information. Moreover, compared to SFTF, TBF has a smaller size and better sparsity while maintaining identical collinearity.
	
	\item We propose MaskDETR-Reg, a Mask-augmented detection Transformer (DETR) for the regression module, to process angle-delay domain information in TBF for localization. According to the mask fusion strategy in the backbone, MaskDETR-Reg can be divided into two variants: early-concat fusion MaskDETR (EC-MaskDETR), which achieves faster convergence, and dual-embedding fusion MaskDETR (DE-MaskDETR), which provides higher accuracy.
	
	\item We propose a fusion-enhanced Transformer for direction classification (Fusion-TDC) module to process Doppler information in TBF for orientation awareness. The estimation results from MaskDETR-Reg serve as prior information to assist Fusion-TDC. MaskDETR-Reg and Fusion-TDC together form LOA-Net to realize localization and orientation awareness simultaneously.
\end{itemize}

The rest of this paper is summarized as follows. In Section \ref{sec:system model}, the system model is constructed. In Section \ref{sec:Fingerprint Features}, we introduce the detailed features of TBF. In Section \ref{sec:LOA architecture}, the architecture of LOA is introduced. Section \ref{sec:simulation results} and Section \ref{sec:conclusion} show the simulation results and conclusions, respectively.

\textit{Notations:} Lowercase $x$, lowercase bold $\mathbf{x}$, uppercase bold $\mathbf{X}$, and calligraphy $\mathcal{X}$ represent scalar, vector, matrix, and tensor, respectively. We use superscripts $(\cdot)^T$, $(\cdot)^H$, and $(\cdot)^*$ to denote the transpose, conjugate-transpose, and conjugate operations, respectively. $\left[\cdot\right]$ represents the operation of fetching elements. For example, $\left[\mathbf{x}\right]_i$, $\left[\mathbf{X}\right]_{i,j}$ and $\left[\mathcal{X}\right]_{i,j,k}$ represent the $i\text{th}$ element of vector $\mathbf{x}$, the $(i,j)\text{th}$ element of matrix $\mathbf{X}$, and the $(i,j,k)\text{th}$ element of the tensor $\mathcal{X}$, respectively, where the element indices start with $0$. $\left [  \cdot \right ] _{i_0,i_1, \cdots ,i_{k-1},:,i_{k+1},\cdots ,i_{\mathrm{M}}}$ is the vector of $I_k$ elements of the $k\text{th}$ dimension of the tensor with the indices $i_0,i_1, \cdots ,i_{k-1},i_{k+1},\cdots ,i_{\mathrm{M}}$ on the other dimensions fixed. $\mathbb{R}^{M\times \cdots \times N}$ and $\mathbb{C}^{M\times \cdots \times N}$ denote the $M\times \cdots \times N$ dimensional real and complex vector spaces, respectively. We use $\bar{\imath} = \sqrt{-1}$ to denote the imaginary unit. $\left\lfloor x\right\rfloor$ denotes the downward rounding of $x$. The notation $\triangleq $ is used for definitions. $\left\langle b\right\rangle_a$ denotes the modulo operation of $b$ mod $a$. We adopt $\mathbf{I}_{N}$ to denote the $N\times N$ dimensional identity matrix. We adopt $\mathbf{0}$ to denote the zero vector, matrix, or tensor. $\odot$, $\otimes$, $\bullet$, $\circ_{n}$ and $\diamond_{n}$ denote the Hadamard product, Kronecker product, outer product, $n$ mode product and $n$-dimensional Einstein product, respectively. $\mathbb{E}\{\cdot\}$ denotes the expectation operation. $\text{Diag}\{\cdot\}$ and $\text{Tr}\{\cdot\}$ represent taking the main diagonal elements and the trace operation, respectively. $\delta (\cdot)$ denotes the delta function. $\lim_{a\to\infty}$ denotes the limit.

\section{System Model}\label{sec:system model}

In this section, we consider a massive MIMO-OFDM system model and obtain a multidimensional fingerprint containing multiple multipath information in the TB domain.

\subsection{System Configuration}

We consider a single-cell time-division duplex (TDD) massive MIMO-OFDM system with a centrally located BS and $K$ UTs. We use $\mathfrak{K}=\{0,1,\cdots,K-1\}$ to represent the set of all UTs. The BS is equipped with a uniform planar array (UPA) placed in the X-Z plane, which has $M_{\mathrm{r}}$ antennas in each row and $M_{\mathrm{c}}$ antennas in each column, and a total of $A=M_{\mathrm{r}}M_{\mathrm{c}}$ antennas. The spacings between two adjacent antennas in each row and each column are $d_{\mathrm{r}}$ and $d_{\mathrm{c}}$, respectively. Each UT is equipped with a single antenna and is distributed randomly within the cell. We adopt OFDM modulation with the number of subcarriers being $N_{\mathrm{c}}$ and the length of sampling interval being $T_{\mathrm{s}}$. We set the length of the cyclic prefix as $N_{\mathrm{g}}$, then the subcarrier interval $\Delta f=\frac{1}{N_{\mathrm{c}}T_{\mathrm{s}}}$ and the length of each OFDM symbol $T_{\mathrm{sym}}=(N_{\mathrm{c}}+N_{\mathrm{g}})T_{\mathrm{s}}$. The wavelength of the carrier is $\lambda_\mathrm{s}$. In the time domain, we have $N_{\mathrm{f}}$ slots in each frame, and each slot is divided into $N_{\mathrm{s}}$ OFDM symbols. That is, each frame contains $N_{\mathrm{t}}=N_{\mathrm{f}}N_{\mathrm{s}}$ OFDM symbols. For convenience, we still follow the method in \cite{DSLS2023,JTXG2024}, combining the current slot and the previous $(N_{\mathrm{f}}-1)$ slots into a temporary frame for localization.

In this work, we assume that all UTs are located in the far-field region of the UPA and that perfect CSI is available at the BS. Notably, in practical applications, an estimated CSI must first be obtained through channel estimation before fingerprint localization. Due to the effect of channel estimation error, the CSI used for localization is noise-contaminated. The detailed procedure for TB domain channel estimation in multi-UT scenarios under TDD mode can be referred to \cite{JTXG2024}.

\subsection{Tensor Channel Model}

Consider the scenario of the uplink, and the signal sent by each UT goes through multiple paths to the BS. The frame structure and massive MIMO-OFDM system diagram are shown in Fig. \ref{fig:system_diagram}.

We begin by focusing on the scenario within a single OFDM symbol. The signal transmitted by the $k$-th UT is received by the BS after traversing $P$ propagation paths, which include reflections and scattering. We use $\mathfrak{P}=\{0,1,\cdots,P-1\}$ to represent the set of all paths. These paths are assumed to be wide-sense stationary and independent of each other, and the propagation environment may contain a line-of-sight (LOS) path or consist entirely of NLOS paths.

\begin{figure}[!t]
	\centering
	\includegraphics[width=3in]{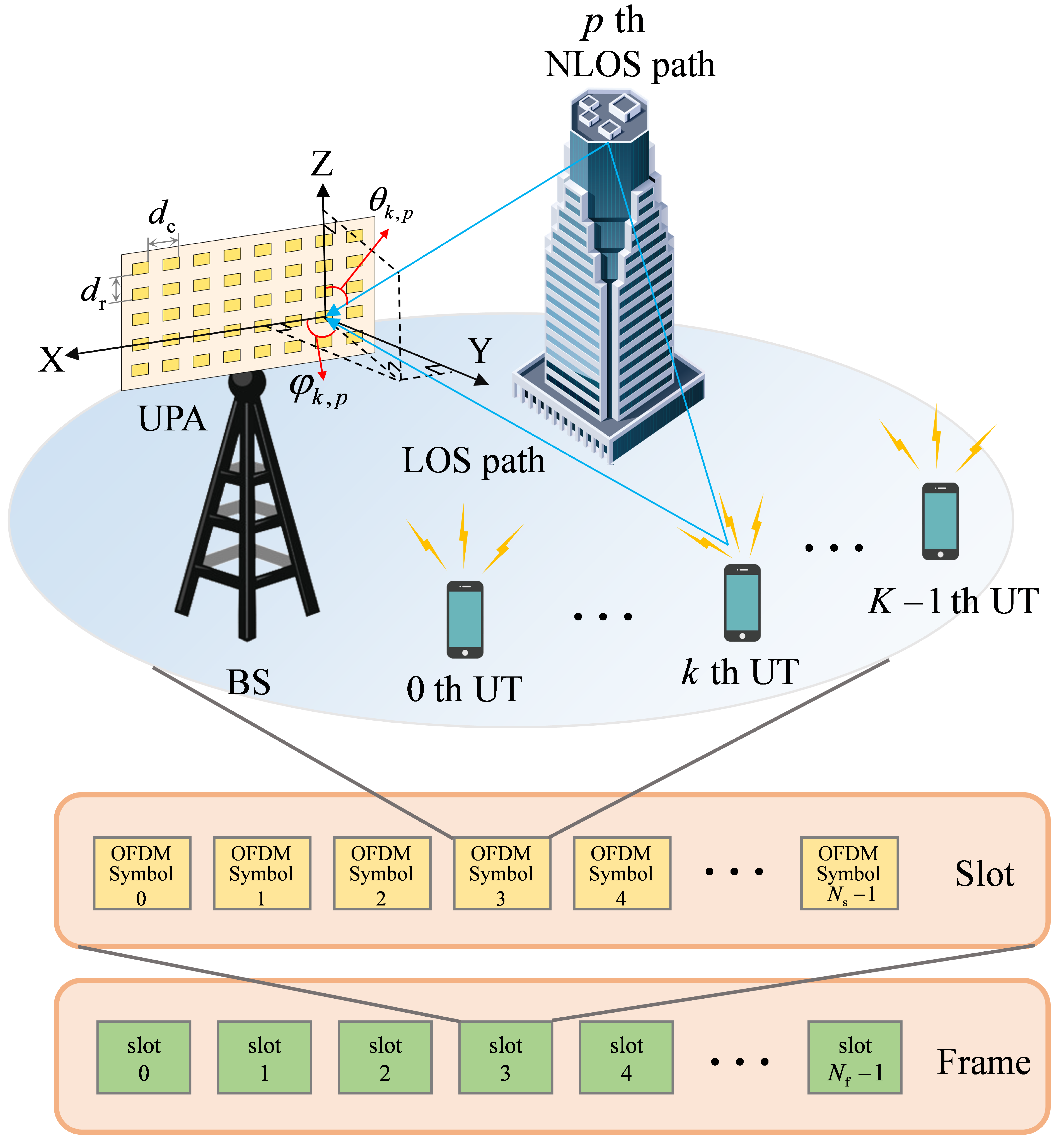}
	\caption{Frame structure and massive MIMO-OFDM system diagram.}
	\label{fig:system_diagram}
\end{figure}

In Fig. \ref{fig:system_diagram}, we decompose the DOA of the $p$-th path into its elevation component $\theta_{k,p}\in[0,\pi]$ in the vertical direction and its azimuth component $\varphi_{k,p}\in[0,\pi]$ in the horizontal direction. The array response vector for each column and row of antennas can be expressed as follows
\begin{equation}
	\mathbf{f}^{c}(\theta_{k,p})=\left[1,e^{-\bar{\imath}2\pi\frac{d_{\mathrm{r}}}{\lambda_c}\cos \theta_{k,p}},\cdots,e^{-\bar{\imath}2\pi\frac{d_{\mathrm{r}}}{\lambda_c}(M_{\mathrm{c}}-1)\cos\theta_{k,p}}\right]^T,
\end{equation}
\begin{equation}
	\begin{aligned}
		\mathbf{f}^{r}(\theta_{k,p},\varphi_{k,p})=&\left[1,e^{-\bar{\imath}2\pi\frac{d_{\mathrm{c}}}{\lambda_c}\sin\theta_{k,p}\cos\varphi_{k,p}},\right.\\
		&\left.\cdots,e^{-\bar{\imath}2\pi\frac{d_{\mathrm{c}}}{\lambda_c}(M_{\mathrm{r}}-1)\sin\theta_{k,p}\cos\varphi_{k,p}}\right]^T,
	\end{aligned}	
\end{equation}
where $\mathbf{f}^{c}(\theta_{k,p})\in\mathbb{C}^{M_{\mathrm{c}}\times 1}$ and $\mathbf{f}^{r}(\theta_{k,p},\varphi_{k,p})\in\mathbb{C}^{M_{\mathrm{r}}\times 1}$. The array response vector of UPA is obtained by combining the phases in both directions:
\begin{equation}\label{eq:array_vector}
	\mathbf{f}^{\phi}(\theta_{k,p},\varphi_{k,p})=\mathbf{f}^{c}(\theta_{k,p})\otimes \mathbf{f}^{r}(\theta_{k,p},\varphi_{k,p})\in\mathbb{C}^{A\times1},
\end{equation}
where $\left[\mathbf{f}^{\phi}(\theta_{k,p},\varphi_{k,p})\right]_{i}$ is the phase factor of the antenna in row $\lfloor \frac{i}{M_{\mathrm{r}}} \rfloor$ and column $\left\langle i\right\rangle_{M_{\mathrm{r}} }$. Denote the Doppler frequency on the $p$-th path of the $k$-th UT as $\nu_{k,p}$, and the phase factor is $f^{\nu}(t) = e^{\bar{\imath}2\pi\nu_{k,p}t}$. The delay of the $p$-th path to the BS is represented as $\tau_{k,p}$. The channel response vector for the $k$-th UT during one OFDM symbol is represented as
\begin{equation}
	\mathbf{h}_{k}(t,\tau)\!=\!\!\sum_{p=0}^{P-1}\!\beta_{k,p}\mathbf{f}^{\phi}(\theta_{k,p},\varphi_{k,p})f^{\nu}(t)\delta(\tau-\tau_{k,p})\!\in\!\mathbb{C}^{A\times 1},
\end{equation}
where $\beta_{k,p}$ represents the gain of the $k$-th UT on the $p$-th path, which follows a complex Gaussian distribution $\mathcal{CN}(0,\sigma^2_{k,p})$ with variance $\sigma^2_{k,p}$. In a frame consisting of $N_{\mathrm{t}}$ OFDM symbols, we assume that the channel state will vary between symbols but will remain constant within an individual symbol. Then the phase factor on the $n$-th OFDM symbol (i.e., $t\in [nT_{\mathrm{sym}},(n+1)T_{\mathrm{sym}})$) is
\begin{equation}\label{eq:doppler factor}
	f^{\nu}(t)\approx f^{\nu}_{k,p,n}=e^{\bar{\imath}2\pi\nu_{k,p}nT_{\mathrm{sym}}}.
\end{equation}
Assume that the duration of the cyclic prefix exceeds the delay of all UT paths, which means that $N_{\mathrm{g}}T_{\mathrm{s}}>\max_{k\in\mathfrak{K},p\in\mathfrak{P}} \left\{\tau_{k,p}\right\}$. The phase factor on the $c$-th subcarrier is represented as
\begin{equation}\label{eq:phase_factor}
	f^{\tau}_{k,p,c}=e^{-\bar{\imath} 2\pi c\tau_{k,p}\Delta f}.
\end{equation}
For UT $k$, the channel response vector on the $c$-th subcarrier of the $n$-th OFDM symbol is represented as \cite{DSLS2023,JTXG2024}
\begin{align}
	\mathbf{h}^{\mathrm{SFT}}_{k,c,n}\triangleq&\int\mathbf{h}_{k}(nT_{\mathrm{sym}},\tau)f^{\tau}_{k,p,c}\mathrm{d}\tau\nonumber\\
	=&\sum_{p=0}^{P-1}\beta_{k,p}\mathbf{f}^{\phi}(\theta_{k,p},\varphi_{k,p})f^{\tau}_{k,p,c}f^{\nu}_{k,p,n}\in\mathbb{C}^{A\times 1}.
\end{align}

In the proposed channel model, we incorporate information from the space, time and frequency domains, corresponding to $\mathbf{f}^{\phi}(\theta_{k,p},\varphi_{k,p})$, $f^{\tau}_{k,p,c}$ and $f^{\nu}_{k,p,n}$, respectively. In contrast to the SFT domain channel models in the previous paper \cite{JTXG2024}, we employ UPA to extend the space domain from 2D to 3D, enriching the fingerprint information available. Following (\ref{eq:array_vector}), we represent the frequency domain steering vector $\mathbf{f}^{\tau}(\tau_{k,p})$ and time domain steering vector $\mathbf{f}^{\nu}(\nu_{k,p})$ as
\begin{equation}
	\mathbf{f}^{\tau}(\tau_{k,p})=\left[f^{\tau}_{k,p,0},f^{\tau}_{k,p,1},\cdots,f^{\tau}_{k,p,N_{\mathrm{c}}-1}\right]^T,
\end{equation}
\begin{equation}
	\mathbf{f}^{\nu}(\nu_{k,p})=e^{ \bar{\imath}2\pi n_{\mathrm{T}}N_{\mathrm{s}} \nu_{k,p} T_{\mathrm{sym}} }\cdot\left[f^{\nu}_{k,p,0},f^{\nu}_{k,p,1},\cdots,f^{\nu}_{k,p,N_{\mathrm{t}}-1}\right]^T,
\end{equation}
where $\mathbf{f}^{\tau}(\tau_{k,p})\in\mathbb{C}^{N_{\mathrm{c}}\times 1}$, $\mathbf{f}^{\nu}(\nu_{k,p})\in\mathbb{C}^{N_{\mathrm{t}}\times 1}$, and $n_{\mathrm{T}}$ is the first OFDM symbol of this frame. We specify the range of multipath delay and Doppler frequency \cite{JTXG2024}:
\begin{align}
	\tau_{k,p}&\in\left[0,\frac{1}{N_{\mathrm{c}}\Delta f},\frac{2}{N_{\mathrm{c}}\Delta f},\cdots,\frac{N_{\mathrm{c}}-1}{N_{\mathrm{c}}\Delta f}\right],\nonumber\\ 
	\nu_{k,p}&\in\left[\frac{0-N_{\mathrm{f}}/2}{N_{\mathrm{t}}T_{\mathrm{sym}}},\frac{1-N_{\mathrm{f}}/2}{N_{\mathrm{t}}T_{\mathrm{sym}}},\cdots,\frac{N_{\mathrm{f}}-1-N_{\mathrm{f}}/2}{N_{\mathrm{t}}T_{\mathrm{sym}}}\right].\nonumber
\end{align}
With these three steering vectors, we can represent the SFT domain tensor channel model of UT $k$ as follows.
\begin{equation}\label{eq:SFT channel}
	\mathcal{H}^{\mathrm{SFT}}_k=\sum_{p=0}^{P-1}\beta_{k,p}\mathbf{f}^{\phi}(\theta_{k,p},\varphi_{k,p})\bullet\mathbf{f}^{\tau}(\tau_{k,p})\bullet\mathbf{f}^{\nu}(\nu_{k,p}),
\end{equation}
where the outer product of $\mathcal{A}\in\mathbb{C}^{I_0 \times I_1 \times \cdots \times I_M}$ and $\mathcal{B}\in\mathbb{C}^{J_0 \times J_1 \times \cdots \times J_N}$ is $\mathcal{A}\bullet\mathcal{B}\in\mathbb{C}^{I_0 \times I_1 \times \cdots \times I_M \times J_0 \times J_1 \times\cdots\times J_N}$, defined as
\begin{equation}
	\left[\mathcal{A}\bullet\mathcal{B}\right]_{i_0,\cdots,i_M,j_0,\cdots,j_N}=\left[\mathcal{A}\right]_{i_0,\cdots,i_M}\left[\mathcal{B}\right]_{j_0,\cdots,j_N}.
\end{equation}
The tensor $\mathcal{H}^{\mathrm{SFT}}_k\in\mathbb{C}^{A\times N_{\mathrm{c}}\times N_{\mathrm{t}}}$ embodies channel information across the space, frequency, and time domains. But the SFT domain channel tensor comprises numerous elements, leading to a significant computational load. To address this, we transform the SFT domain channel tensor into the TB domain and utilize the sparse TB domain channel tensor to tackle the localization problem \cite{JTXG2024}. Then, we define the transform matrices corresponding to the three beam domains:
\begin{align}
	\left[\mathbf{W}^{\phi}_{M}\right]_{i,j}&\triangleq\frac{1}{\sqrt{M}}e^{-\bar{\imath}2\pi\frac{i(2j-M)}{2M}},\label{eq:angle transform matrices}\\
	\left[\mathbf{W}^{\tau}_{N_{\mathrm{c}},N_{\mathrm{g}}}\right]_{i,j}&\triangleq\frac{1}{\sqrt{N_{\mathrm{c}}}}e^{-\bar{\imath}2\pi\frac{ij}{N_{\mathrm{c}}}},\label{eq:delay transform matrices}\\
	\left[\mathbf{W}^{\nu}_{N_{\mathrm{t}},N_{\mathrm{f}}}\right]_{i,j}&\triangleq \frac{1}{\sqrt{N_{\mathrm{t}}}}e^{ \bar{\imath}2\pi(n_{\mathrm{T}}+\frac{i}{N_{\mathrm{s}}})\frac{(2j-N_{\mathrm{f}})}{2N_{\mathrm{f}}}}\label{eq:dopp transform matrices},
\end{align}
where $M$ denote $M_{\mathrm{c}}$ or $M_{\mathrm{r}}$, matrices $\mathbf{W}^{\phi}_{M}\in\mathbb{C}^{M\times M}$, $\mathbf{W}^{\tau}_{N_{\mathrm{c}},N_{\mathrm{g}}}\in\mathbb{C}^{N_{\mathrm{c}}\times N_{\mathrm{g}}}$ and $\mathbf{W}^{\nu}_{N_{\mathrm{t}},N_{\mathrm{f}}}\in\mathbb{C}^{N_{\mathrm{t}}\times N_{\mathrm{f}}}$. The detailed derivation of the transformation can be found in \cite{LYXG2015,LYXG2016,JTXG2024,CWXY2021}.

Rerepresent (\ref{eq:SFT channel}) in terms of TB domain channel $\mathcal{H}^{\mathrm{TB}}_k$,
\begin{equation}\label{eq:TB to SFT}
	\begin{aligned}
		\mathcal{H}^{\mathrm{SFT}}_k=&\mathbf{W}^{\nu}_{N_{\mathrm{t}},N_{\mathrm{f}}}\!\circ_{3}\!\Big(\!\mathbf{W}^{\tau}_{N_{\mathrm{c}},N_{\mathrm{g}}}\!\circ_{2}\!\big((\mathbf{W}^{\phi}_{M_{\mathrm{c}}}\otimes\mathbf{W}^{\phi}_{M_{\mathrm{r}}})\!\circ_{1}\!\mathcal{H}^{\mathrm{TB}}_k\big)\Big)\\
		&\times \sqrt{M_{\mathrm{c}}M_{\mathrm{r}}N_{\mathrm{c}}N_{\mathrm{t}}}\\
		=&\sqrt{M_{\mathrm{c}}M_{\mathrm{r}}N_{\mathrm{c}}N_{\mathrm{t}}}\cdot\mathcal{W}\diamond_{3}\mathcal{H}^{\mathrm{TB}}_k,
	\end{aligned}
\end{equation}
where $\mathcal{H}^{\mathrm{TB}}_k\in\mathbb{C}^{A\times N_{\mathrm{g}}\times N_{\mathrm{f}}}$ is the TB domain channel of the $k$-th UT whose $(i,j,l)$-th element is $\sum_{p=0}^{P-1}\beta_{k,p}$. The triple-beam domain refers to the transformed angle, delay, and Doppler beam domains, which correspond one-to-one with the space, frequency, and time domains of the SFT domain. In the case of $n$ mode product between $\mathbf{A}\in\mathbb{C}^{I\times J_n}$ and $\mathcal{B}\in\mathbb{C}^{J_0\times J_1\times \cdots \times J_n \times \cdots \times J_N}$, we have $\mathbf{A}\circ_{n}\mathcal{B}\in\mathbb{C}^{J_0\times J_1\times \cdots \times I \times \cdots \times J_N}$ and each of its elements
\begin{equation}
	\left[\mathbf{A}\circ_{n}\mathcal{B}\right]_{J_0,J_1,\cdots,i,\cdots,J_N}=\sum_{j_n=0}^{J_n-1}\left[\mathbf{A}\right]_{i,j_n}\left[\mathcal{B}\right]_{J_0,J_1,\cdots,j_n,\cdots,J_N}.
\end{equation}
We define $\mathcal{W}\in\mathbb{C}^{A\times N_{\mathrm{c}}\times N_{\mathrm{t}} \times A \times N_{\mathrm{g}} \times N_{\mathrm{f}}}$ as a triple-beam tensor, where each element is denoted by 
\begin{align}
	&\left[\mathcal{W}\right]_{a,c,n,a',c',n'}\nonumber\\
	\triangleq&\left[\mathbf{W}^{\phi}_{M_{\mathrm{c}}}\otimes\mathbf{W}^{\phi}_{M_{\mathrm{r}}}\right]_{a,a'}\left[\mathbf{W}^{\tau}_{N_{\mathrm{c}},N_{\mathrm{g}}}\right]_{c,c'}\left[\mathbf{W}^{\nu}_{N_{\mathrm{t}},N_{\mathrm{f}}}\right]_{n,n'}.
\end{align}
For the Einstein product of $\mathcal{W}$ and $\mathcal{H}^{\mathrm{TB}}_k$, we have $\mathcal{W}\diamond_{3}\mathcal{H}^{\mathrm{TB}}_k\in\mathbb{C}^{A\times N_{\mathrm{c}}\times N_{\mathrm{t}}}$ and each of its elements
\begin{equation}
	\begin{aligned}
		&\left[\mathcal{W}\diamond_{3}\mathcal{H}^{\mathrm{TB}}_k\right]_{i_1,i_2,i_3}\\
		=&\sum_{j_1=0}^{A-1}\sum_{j_2=0}^{N_{\mathrm{g}}-1}\sum_{j_3=0}^{N_{\mathrm{f}}-1}\left[\mathcal{W}\right]_{i_1,i_2,i_3,j_1,j_2,j_3}\left[\mathcal{H}^{\mathrm{TB}}_k\right]_{j_1,j_2,j_3}.
	\end{aligned}
\end{equation}

The multipath channel in wireless communication is influenced by the position of UT and is determined by the scattering environment between UT and BS. Variations in the scattering environment will result in distinct channel fingerprints. The tensor $\mathcal{H}^{\mathrm{SFT}}_k$ contains a wealth of channel information, including RSS, DOA, TOA, and Doppler frequency. Different positions of UT will yield different tensors $\mathcal{H}^{\mathrm{SFT}}_k$. To achieve more stable channel information, we introduce the second-order statistics of $\mathcal{H}^{\mathrm{SFT}}_k$:
\begin{equation}
	\mathcal{X}_{k}\triangleq\mathbb{E}\left\{\mathcal{H}^{\mathrm{SFT}}_k\bullet (\mathcal{H}^{\mathrm{SFT}}_k)^{*}\right\},
\end{equation}
which represents the SFT domain channel covariance tensor. As a kind of statistical CSI, $\mathcal{X}_{k}\in\mathbb{C}^{A\times N_{\mathrm{c}}\times N_{\mathrm{t}} \times A \times N_{\mathrm{c}} \times N_{\mathrm{t}}}$ mitigates the effects of small-scale fading and remains relatively consistent over time, making it suitable for a fingerprint, referred to as SFTF. Since tensor $\mathcal{X}_{k}$ has a high dimension, its computation is substantial. This is the reason why we utilize the sparsity of the TB domain channel to simplify the calculation, and the sparsity of $\mathcal{H}^{\mathrm{TB}}_k$ has been proved in \cite{JTXG2024}. We define the following TBF,
\begin{equation}\label{eq:TBF}
	\mathcal{F}_{k}\triangleq\mathbb{E}\left\{\mathcal{H}^{\mathrm{TB}}_k\odot (\mathcal{H}^{\mathrm{TB}}_k)^{*}\right\},
\end{equation}
which is also the TB domain channel power tensor. This TBF $\mathcal{F}_{k}\in\mathbb{C}^{A\times N_{\mathrm{g}}\times N_{\mathrm{f}}}$ serves as a stable fingerprint that incorporates multiple channel features. From a physical perspective, the TBF contains information from the angle, delay, and Doppler beam domains; therefore, we refer to it as the triple-beam fingerprint. Remarkably, TBF achieves positioning accuracy comparable to the SFTF while being more compact in size. In Section \ref{sec:Fingerprint Features}, we will evaluate the viability of $\mathcal{F}_{k}$ as a fingerprint by demonstrating the distinctiveness of various TBFs, and collinearity among TBF and among SFTF.

\section{Fingerprint Features in TB Domain}\label{sec:Fingerprint Features}

This section provides a detailed analysis of the features of TBF. We first prove that TBF is an effective fingerprint, capturing multiple channel information and exhibiting strong discriminability. This highlights the potential of TBF for accurate positioning. Then, we establish the collinearity between sparse TBF and SFTF. This indicates that TBF achieves comparable localization performance despite its reduced dimensionality.

\subsection{Discriminability of TBF}

The accuracy of fingerprint location is linked to the discriminability of fingerprints. TBF encompasses various CSI in the angle, delay, and Doppler domains. In contrast to fingerprints that rely on RSS, DOA, or TOA, TBF offers enhanced resolution, resulting in more pronounced differentiation of fingerprints generated by UT at various locations. For convenience, we define unit vector $\mathbf{\Lambda} _{L}^{a}\in\mathbb{R}^{L\times 1}$,
\begin{equation}
	[\mathbf{\Lambda} _{L}^{a}]_{i}=\begin{cases}
		1, & i=a \\
		0, & \text{else}
	\end{cases},
\end{equation}
where $L$ is the length of $\mathbf{\Lambda} _{L}^{a}$, and $a\in[0,1,\dots,L-1]$. The theorem presented below demonstrates the high discriminability of TBF.

\begin{theorem}\label{theorem:1}
	For a massive MIMO-OFDM system, when all dimensions in the TB domain approach infinity, the channel power of each path will be concentrated at a certain location, which can be expressed as
	\begin{align}\label{eq:theorem 1}
		&\lim_{M_{\mathrm{c}},M_{\mathrm{r}},N_{\mathrm{c}},N_{\mathrm{t}}\to\infty}\Big(\mathcal{F}_{k}\nonumber\\
		-&\sum_{p=0}^{P-1}\!\sigma^2_{k,p}\!\!\left(\!\mathbf{\Lambda} _{M_{\mathrm{c}}}^{\breve{c}_{k,p}}\!\!\otimes\!\mathbf{\Lambda} _{M_{\mathrm{r}}}^{\breve{r}_{k,p}}\!\right)\!\bullet\!\mathbf{\Lambda} _{N_{\mathrm{g}}}^{(\frac{\tau_{k,p}}{T_{\mathrm{s}}})}\!\!\bullet\!\mathbf{\Lambda} _{N_{\mathrm{f}}}^{(N_{\mathrm{t}}\nu_{k,p} T_{\mathrm{sym}}\!+\!\frac{N_{\mathrm{f}}}{2})}\Big)=\mathbf{0},
	\end{align}where $\frac{\tau_{k,p}}{T_{\mathrm{s}}}$ and $N_{\mathrm{t}}\nu_{k,p} T_{\mathrm{sym}}$ can be thought of as integers since $T_{\mathrm{s}}$ and $T_{\mathrm{sym}}$ are small and $N_{\mathrm{t}}$ and $N_{\mathrm{f}}$ are large, and
	\begin{align}
		\breve{c}_{k,p}=&\frac{M_{\mathrm{c}}d_{\mathrm{r}}}{\lambda_c}\cos\theta_{k,p}+\frac{M_{\mathrm{c}}}{2},\\
		\breve{r}_{k,p}=&\frac{M_{\mathrm{r}}d_{\mathrm{c}}}{\lambda_c}\sin\theta_{k,p}\cos\varphi_{k,p}+\frac{M_{\mathrm{r}}}{2},
	\end{align}
	which may be thought of integers when $M_{\mathrm{c}}$ and $M_{\mathrm{r}}$ are large. If only one dimension in the TB domain approaches infinity, the channel power tensor will be concentrated in some specific positions of this dimension.
\end{theorem}
\begin{IEEEproof}
	See Appendix \ref{app:Theorem 1}.
\end{IEEEproof}

Theorem \ref{theorem:1} demonstrates that as $M_{\mathrm{c}}$, $M_{\mathrm{r}}$, $N_{\mathrm{c}}$ and $N_{\mathrm{t}}$ approach infinity, for the TBF of the $k$-th UT, the channel energy of the $p$-th multipath becomes focused on the $(\breve{c}_{k,p}M_{\mathrm{r}}+\breve{r}_{k,p})$-th angle direction, the delay $\frac{\tau_{k,p}}{T_{\mathrm{s}}}$, and the $(N_{\mathrm{t}}\nu_{k,p} T_{\mathrm{sym}}\!+\!\frac{N_{\mathrm{f}}}{2})$-th OFDM symbol. In other words, $\left[\mathcal{F}_{k}\right]_{i,j,l}$ corresponds to the average channel power associated with the $i$-th DOA, the $j$-th TOA, and the $l$-th Doppler frequency. A non-zero element in $\mathcal{F}_{k}$ indicates the existence of a signal path arriving at the receiver with the corresponding DOA, TOA, and Doppler frequency. Conversely, $\left[\mathcal{F}_{k}\right]_{i,j,l}=0$ implies that no signal path reaches the receiver under those parameters. This theorem elucidates the physical significance of each element of the TBF, highlighting that TBF is a unique fingerprint linked to its location and environment. Fig. \ref{fig:TB_tensor} shows the 3D structure of TBF and its cross-sectional slices along each dimension. The structure of the fingerprint tensor relies on multipath information set $\mathfrak{S}_{k}=\{\sigma_{k,p},\ \theta_{k,p},\ \varphi_{k,p},\ \tau_{k,p},\ \nu_{k,p}|p=0,\cdots,P-1\}$. Due to the channel environment of each location is different, the multipath information sets of two UTs at diverse positions are also different, i.e., $\mathfrak{S}_{k}\neq\mathfrak{S}_{k'}$. Consequently, the corresponding TBFs will differ as well. Theorem \ref{theorem:1} establishes the discriminability of TBFs and illustrates that TBF can function as a unique identifier for UT locations. 
	
\begin{figure}[h!]
	\centering
	\includegraphics[width=3.45in]{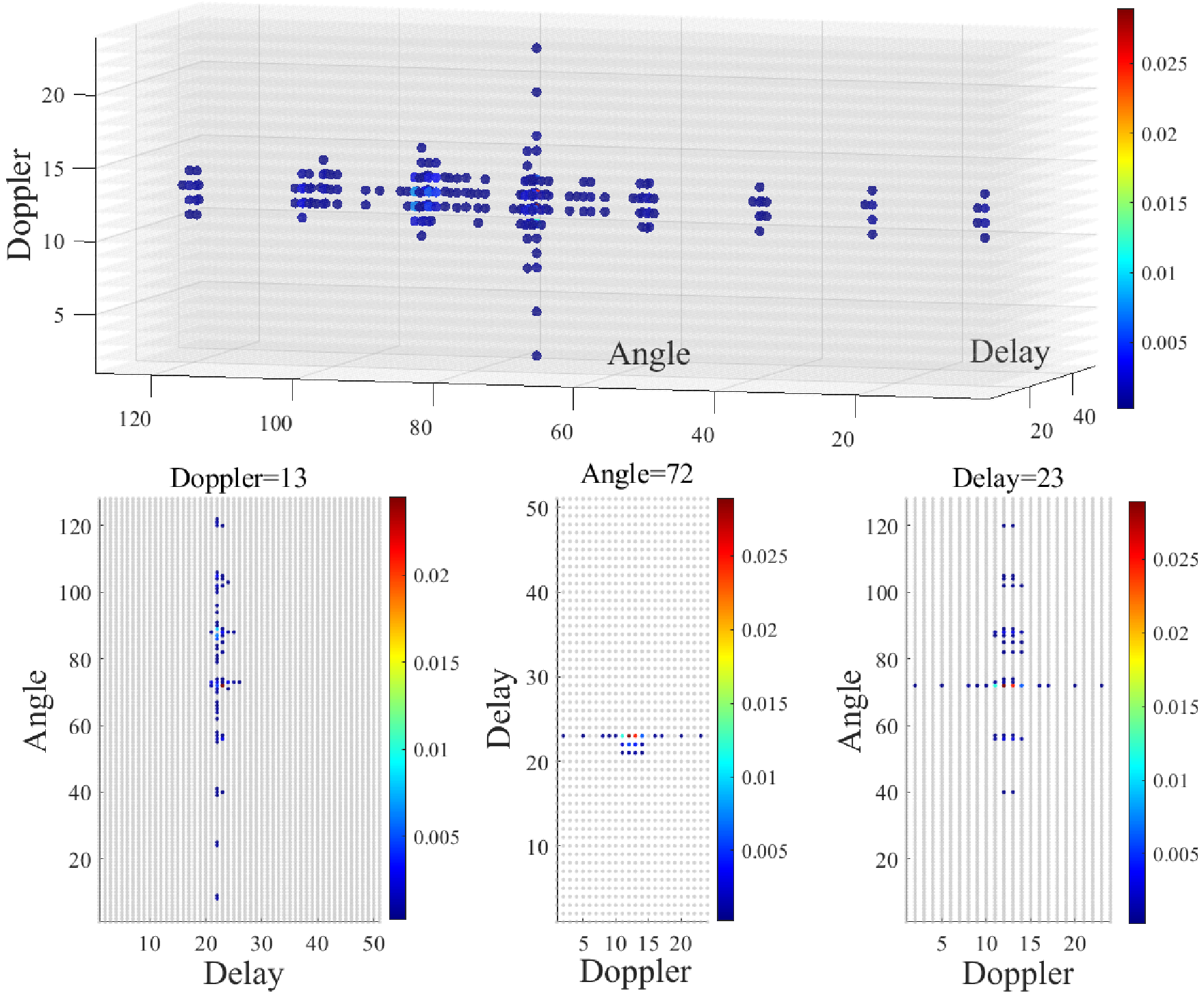}
	\caption{The structure of TBF and its slices in each dimension. $A=128$, $N_{\mathrm{g}}=50$ and $N_{\mathrm{f}}=24$. Values less than $2\times 10^{-4}$ are represented in light gray for visual clarity.}
	\label{fig:TB_tensor}
\end{figure}

Additionally, Theorem \ref{theorem:1} reveals the channel sparsity in the TB domain, which is consistent with \cite{LYXG2015,YZLY20251,JTXG2024}. As $M_{\mathrm{c}}$, $M_{\mathrm{r}}$, $N_{\mathrm{c}}$, and $N_{\mathrm{t}}$ approach infinity, channel energy concentrates solely on specific TB domain positions, indicating the sparsity of TBF. Numerical simulations show TBF maintains sparsity even with large finite dimensions.

We propose a massive MIMO-OFDM tensor channel model and treat stable statistical information $\mathcal{X}_{k}$ as SFTF. Hence, the SFTF encompasses the same information for each multipath component as $\mathcal{H}^{\mathrm{SFT}}_k$, making it a comprehensive and reliable fingerprint. In order to avoid the inherent defects of complex calculation, we transform it into the TB domain and obtain the lightweight TBF. Although Theorem \ref{theorem:1} demonstrates that the TBF retains essential statistical information like SFTF, equating their localization capabilities is cursory. In the next section, we will explore the collinearity of two fingerprints.

\subsection{Collinearity of SFTF and TBF}

The greater the difference in fingerprints at any different positions, the stronger the discriminative ability. Therefore, we define the collinearity of tensors,
\begin{equation}\label{eq:coll define}
	\Xi \left (\mathcal{F}_{k},\mathcal{F}_{k'}\right)=\frac{\text{Tr}\{\mathcal{F}_{k}\bullet\mathcal{F}_{k'}^*\}}{\left \| \mathcal{F}_{k} \right \|_\mathrm{F} \left \| \mathcal{F}_{k'} \right \|_\mathrm{F} },
\end{equation}
where $\left \| \cdot \right \|_\mathrm{F}$ denotes the Frobenius norm. Note the trace operation of tensors, for which we have the following explanation. For $\mathcal{A}\in\mathbb{C}^{I_0\times \cdots I_M\times J_0\times \cdots J_M}$, we refer to the elements $[\mathcal{A}]_{i_0,\cdots,i_M,j_0,\cdots,j_M}$ satisfying condition $i_0=j_0,\cdots,i_M=j_M$ as pseudo-diagonal elements and the trace of $\mathcal{A}$ is the sum of these pseudo-diagonal elements. In particular,
\begin{equation}\label{eq:tensor trace}
	\text{Tr}\{\mathcal{A}\}\triangleq\sum_{i_0=0}^{I_0-1}\cdots\sum_{i_M=0}^{I_M-1}\left[\mathcal{A}\right]_{i_0,\cdots,i_M,i_0,\cdots,i_M}.
\end{equation}
Collinearity measures the similarity of fingerprints at different positions, with lower collinearity indicating greater discrimination \cite{CWXY2021,GHCF2013}. Or rather, if the collinearity between fingerprints is lower, it is more favorable for localization, as it facilitates the discrimination of adjacent points and improves resolution. The following theorem explores the collinearity of SFTF and TBF.

\begin{theorem}\label{theorem:2}
	In the context of a massive MIMO-OFDM system, both SFTF and TBF exhibit equivalent collinearity as $M_{\mathrm{c}}$, $M_{\mathrm{r}}$, $N_{\mathrm{c}}$, and $N_{\mathrm{t}}$ approach infinity, indicating that their localization capabilities are identical. Specifically, when $M_{\mathrm{c}}$, $M_{\mathrm{r}}$, $N_{\mathrm{c}}$ and $N_{\mathrm{t}}$ are sufficiently large, $\mathcal{F}_{k}$ and $\mathcal{X}_{k}$ satisfy the following relation,
	\begin{equation}\label{eq:coll equal}
		\Xi\left(\mathcal{F}_{k},\mathcal{F}_{k'}\right)\approx\Xi\left(\mathcal{X}_{k},\mathcal{X}_{k'}\right).
	\end{equation}
\end{theorem}

\begin{IEEEproof}
	See Appendix \ref{app:Theorem 2}.
\end{IEEEproof}

Theorem \ref{theorem:2} shows the collinearities of SFTF and TBF becomes asymptotically equivalent as $M_{\mathrm{c}}$, $M_{\mathrm{r}}$, $N_{\mathrm{c}}$, $N_{\mathrm{t}}$ approach infinity, and are close when $M_{\mathrm{c}}$, $M_{\mathrm{r}}$, $N_{\mathrm{c}}$ and $N_{\mathrm{t}}$ are large. Building on this foundation and referencing Theorem \ref{theorem:1}, we can conclude that both possess equivalent location information and exhibit the same capability for location discriminability. However, TBF offers advantages such as smaller size and excellent sparsity in practical applications. Consequently, this paper selects TBF as the input for the network model.

\section{Localization And Orientation Awareness Network}\label{sec:LOA architecture} 

Drawing upon the comprehensive environmental information contained in the TBF, we propose the LOA-Net to estimate both the localization and motion direction of a UT. Given practical constraints, we select the Angle-Delay domain information for localization, and the Doppler domain information is exclusively used to estimate the motion direction. Accordingly, the localization functionality of LOA-Net is realized via MaskDETR-Reg, and the motion direction estimation is performed by Fusion-TDC. Fig. \ref{fig:LOA Net} shows the architecture of the LOA-Net. In this section, we introduce each component.

\begin{figure*}[!t]
	\centering
	\includegraphics[width=6.5in]{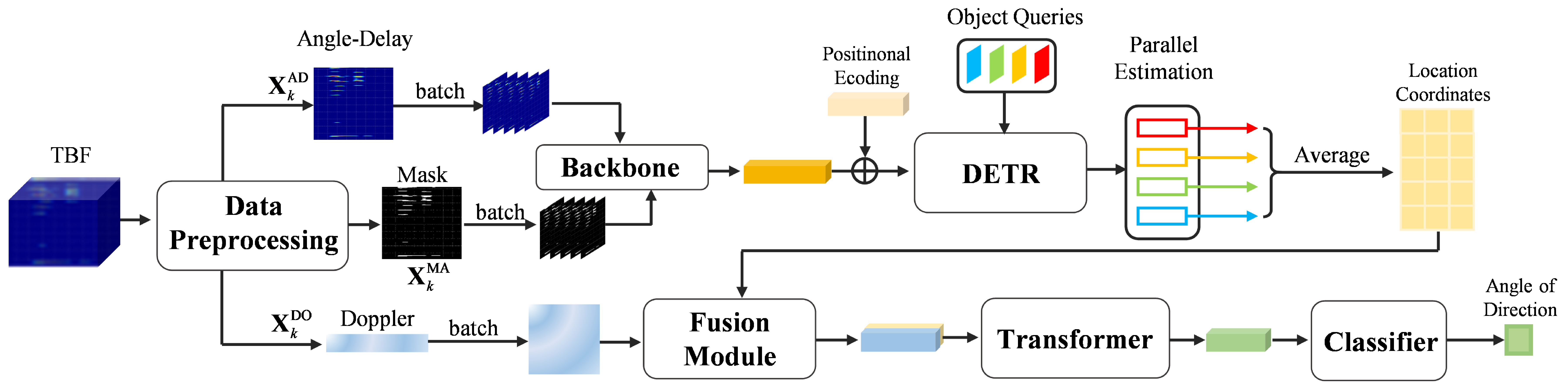}
	\caption{The network architecture of the proposed LOA-Net.}
	\label{fig:LOA Net}
\end{figure*}

\subsection{Data Preprocessing}

As illustrated in Fig. \ref{fig:LOA Net}, $\mathcal{F}_{k}$ separates the Angle-Delay and Doppler domain information after data preprocessing, and simultaneously generates a mask. The Angle-Delay domain data and mask are aggregated into batches and fed as input to MaskDETR-Reg, while the Doppler domain information is similarly batched and processed by Fusion-TDC. 

Since the duration of a frame is very short, we assume that the UT position remains constant for each Doppler frequency in $\mathcal{F}_{k}$, implying $\left[\mathcal{F}_{k}\right]_{:,:,i}\approx\beta^{\mathrm{DO}}_{i,j}\left[\mathcal{F}_{k}\right]_{:,:,j}$. Here $\beta^{\mathrm{DO}}_{i,j}$ is the coefficient related to the Doppler frequencies of the $i$th and $j$th slots. This is also verified numerically by subsequent simulation. Defining the Angle-Delay domain data of the $k$th UT as $\mathbf{X}^{\mathrm{AD}}_{k}\in\mathbb{R}^{A\times N_{\mathrm{g}}}$. By aggregating the third dimension of $\mathcal{F}_{k}$ and normalizing, its Angle-Delay domain information can be extracted,
\begin{equation}\label{eq:AD input}
	\mathbf{X}^{\mathrm{AD}}_{k}=\frac{\sum_{l=0}^{N_{\mathrm{f}}-1}\left[\mathcal{F}_{k}\right]_{:,:,l}}{\sum_{i=0}^{A-1}\sum_{j=0}^{N_{\mathrm{g}}-1}\sum_{l=0}^{N_{\mathrm{f}}-1}\left[\mathcal{F}_{k}\right]_{i,j,l}}.
\end{equation}
Given that each $\left[\mathcal{F}_{k}\right]_{:,:,l}$ shares the extremely similar sparsity pattern, $\sum_{l=0}^{N_{\mathrm{f}}-1}\left[\mathcal{F}_{k}\right]_{:,:,l}$ preserves the Angle-Delay domain information across all Doppler frequencies while enabling dimensionality reduction. Here, $\mathbf{X}^{\mathrm{AD}}_{k}$ has undergone normalization and can be directly input into the MaskDETR-Reg after batch processing.

In computer vision, masks are commonly employed to restrict the focus of the network, thereby mitigating the impact of irrelevant regions on the output. This approach aligns naturally with the attention mechanism in Transformer \cite{HWCX2022}. Defining the mask of the $k$th UT as $\mathbf{X}^{\mathrm{MA}}_{k}\in\mathbb{R}^{A\times N_{\mathrm{g}}}$, we have
\begin{equation}\label{eq:mask input}
	\left[\mathbf{X}^{\mathrm{MA}}_{k}\right]_{i,j}=\left\{\begin{matrix}
		1,& \left[ \mathbf{X}^{\mathrm{AD}}_{k} \right]_{i,j}\ge \gamma \\
		0,& \mathrm{else} 
	\end{matrix}\right.,
\end{equation}
where $\gamma\in [0,1]$ is the threshold parameter, and a larger $\gamma$ means a more restricted focal area.

Similar to the generation of $\mathbf{X}^{\mathrm{AD}}_{k}$, the Doppler domain information $\mathbf{X}^{\mathrm{DO}}_{k}\in\mathbb{R}^{N_{\mathrm{f}}}$ can be extracted by aggregating the first and second dimensions of $\mathcal{F}_{k}$ and then normalizing,
\begin{equation}\label{eq:DO input}
	\mathbf{X}^{\mathrm{DO}}_{k}=\frac{\sum_{i=0}^{A-1}\sum_{j=0}^{N_{\mathrm{g}}-1}\left[\mathcal{F}_{k}\right]_{i,j,:}}{\sum_{i=0}^{A-1}\sum_{j=0}^{N_{\mathrm{g}}-1}\sum_{l=0}^{N_{\mathrm{f}}-1}\left[\mathcal{F}_{k}\right]_{i,j,l}}.
\end{equation}

\subsection{Mask-Augmented DETR for Regression}

According to Theorem \ref{theorem:1}, a bijective correspondence exists between $\mathbf{X}^{\mathrm{AD}}_{k}$ and UT position $\mathbf{Y}^{\mathrm{AD}}_{k}=(x,y,z)\in\mathbb{R}^{3}$. From the perspective of computer vision, $\mathbf{X}^{\mathrm{AD}}_{k}$ can be represented as a single-channel grayscale image, while the estimated position $\hat{\mathbf{Y}}^{\mathrm{AD}}_{k}=(\hat{x},\hat{y},\hat{z})$ constitutes a continuous-valued vector. This formulation naturally frames the localization problem as an end-to-end regression task. The Transformer architecture offers significant advantages over traditional CNNs by overcoming their limited receptive fields \cite{Transformer}. Through its self-attention mechanism and positional encoding, Transformer effectively models relationships among multiple fingerprint features while capturing long-range dependencies—particularly advantageous for processing high-dimensional fingerprints like TBF. DETR adds a multi-query design to the classical transformer, which enables position information extraction from multiple perspectives \cite{DETR}. Based on this, we propose MaskDETR-Reg for solving regression tasks. Formally, the network model can be viewed as a nonlinear function $f(\cdot)$, that is, $\hat{\mathbf{Y}}^{\mathrm{AD}}_{k}=f(\mathbf{X}^{\mathrm{AD}}_{k})$. The detailed architecture is shown in Fig. \ref{fig:MaskDETR}.

\begin{figure*}[!t]
	\centering
	\includegraphics[width=6in]{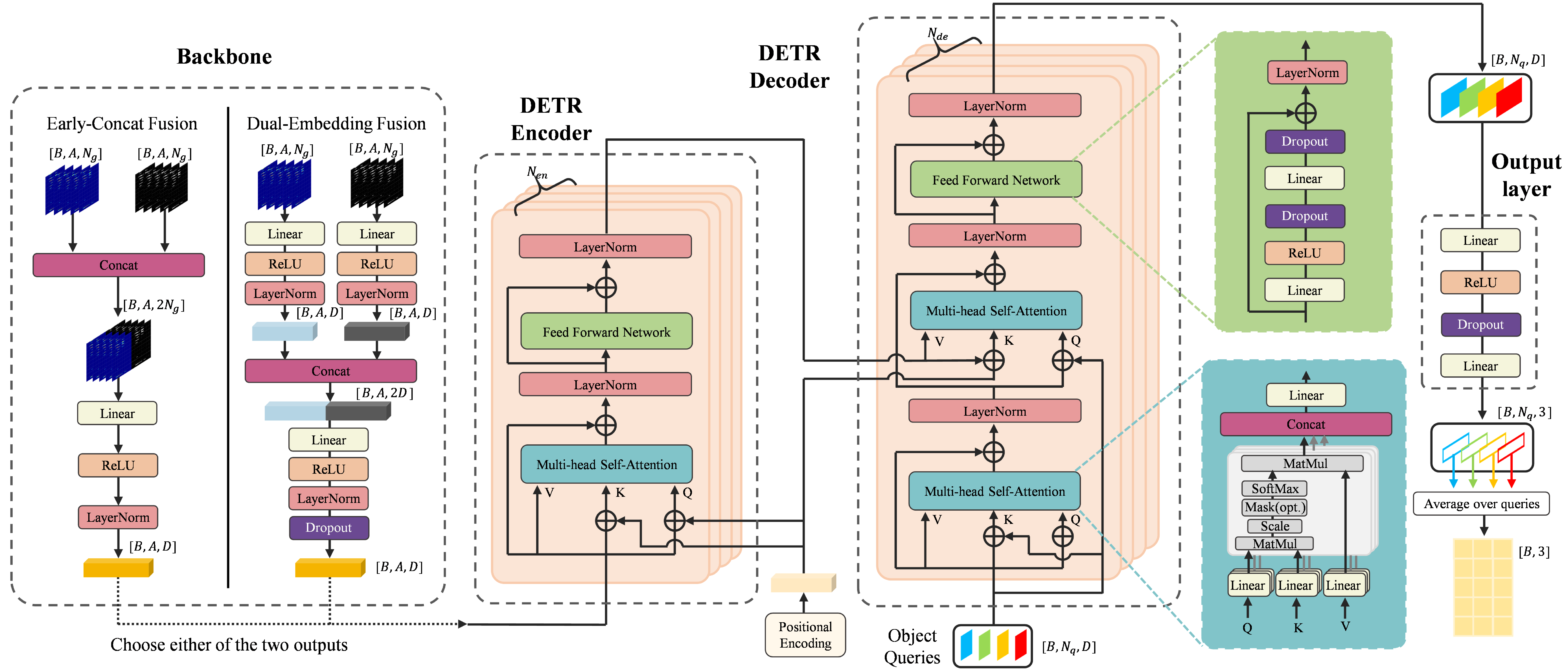}
	\caption{The network architecture of the proposed MaskDETR-Reg.}
	\label{fig:MaskDETR}
\end{figure*}

The proposed model employs a composite loss function $\mathcal{L}_{\mathrm{reg}}$ that combines the localization mean square error (MSE) error and the $L_2$ regularization term to optimize $\hat{\mathbf{Y}}^{\mathrm{AD}}_{k}$. The total loss function is expressed as
\begin{equation}\label{eq:loss reg}
	\mathcal{L}_{\mathrm{reg}}(\boldsymbol{\Theta})=\frac{1}{3N_{\mathrm{tra}}}\sum_{k=1}^{N_{\mathrm{tra}}}\left\| \hat{\mathbf{Y}}^{\mathrm{AD}}_{k}-\mathbf{Y}^{\mathrm{AD}}_{k} \right\|_2^2 + \frac{\lambda}{2}\boldsymbol{\Theta}^T\boldsymbol{\Theta},
\end{equation}
where $N_{\mathrm{tra}}$ is the number of training samples, $\left\|\cdot\right\|_2$ denotes the $\ell_2$-norm, $\frac{\lambda}{2}\boldsymbol{\Theta}^T\boldsymbol{\Theta}$ is the $L_2$ regularization term to prevent overfitting, $\lambda$ is the weight factor of the regularization, and $\boldsymbol{\Theta}$ is the trainable parameters of the network.

The angle-delay domain information and the mask, with batch size $B$, are first processed by the backbone for feature extraction. We propose two fusion strategies: Early-Concat fusion at the input level and Dual-Embedding fusion at the feature level. These strategies effectively combine position information with focal region features. Notably, the mask processing does not discard near-zero values, which is different from \cite{CWXY2021}. In Early-Concat fusion, the input data and its corresponding mask are concatenated along the delay domain dimension. While input level ensures rapid convergence, it lacks deep feature interaction. For Dual-Embedding fusion, the inputs are separately embedded, concatenated, and compressed. The detailed structure is shown in Fig. \ref{fig:MaskDETR} and $D$ denotes the Transformer hidden dimension. The architectures of EC-MaskDETR and DE-MaskDETR differ in their backbone, and the other components remain identical. Both fusion methods have their respective advantages. The EC-MaskDETR is simpler and converges faster, while DE-MaskDETR exhibits stronger feature extraction capabilities and delivers higher localization accuracy. In practical applications, the choice of model can be weighed based on the characteristics of the specific scenario.

The features from the backbone are fed into the DETR encoder and decoder (detailed in \cite{DETR}). In Fig. \ref{fig:MaskDETR}, $N_{\mathrm{en}}$ and $N_{\mathrm{de}}$ denote the numbers of encoder and decoder layers, respectively. $N_{\mathrm{q}}$ denotes the number of object queries. Multi-head Self-Attention and Feed Forward network are the architectures in classical Transformer \cite{Transformer}. The integration of positional encoding and the self-attention mechanism enables the model to capture the inherent structure within the non-sparse regions of the TBF and utilize this information in a focused manner. In our network, the object queries represent $N_{\mathrm{q}}$ parallel estimation results focusing on different input features. The $D$-dimensional features are mapped to xyz coordinates through the output layer, and the average of $N_{\mathrm{q}}$ coordinates is taken as the final result to improve the robustness of estimation.

\subsection{Fusion-Enhanced Transformer for Direction Classification}\label{sec:Fusion TDC}

The Doppler frequency is influenced by multiple factors, and the Doppler domain information in TBF varies depending on position, speed, and motion direction. Using the coordinates estimated by MaskDETR-Reg as prior information, we derive the UT motion direction from $\mathbf{X}^{\mathrm{DO}}_{k}$. Given the constraints imposed by the frame structure in communication systems and considerations of training cost, simultaneously estimating the speed and motion direction of UTs from Doppler power spectra with limited resolution is highly challenging. Therefore, we restrict all UTs to the same speed, simplify the problem as a multi-class classification task, and propose Fusion-TDC. The angular range $[0,2\pi)$ is divided into 16 classes: $[\frac{-\pi}{16},\frac{\pi}{16}),[\frac{\pi}{16},\frac{3\pi}{16}),\cdots,[\frac{29\pi}{16},\frac{31\pi}{16})$. The direction of $\mathbf{X}^{\mathrm{DO}}_{k}$ is $y^{\mathrm{DO}}_{k}$, and its One-Hot encoding is $\mathbf{Y}^{\mathrm{DO}}_{k}\in\mathbb{R}^{16}$. Fusion-TDC is viewed as a nonlinear function $g(\cdot)$, such that $\hat{y}^{\mathrm{DO}}_{k}=g(\mathbf{X}^{\mathrm{DO}}_{k})$. The detailed architecture is shown in Fig. \ref{fig:Transformer_class}.

\begin{figure}[!t]
	\centering
	\includegraphics[width=3in]{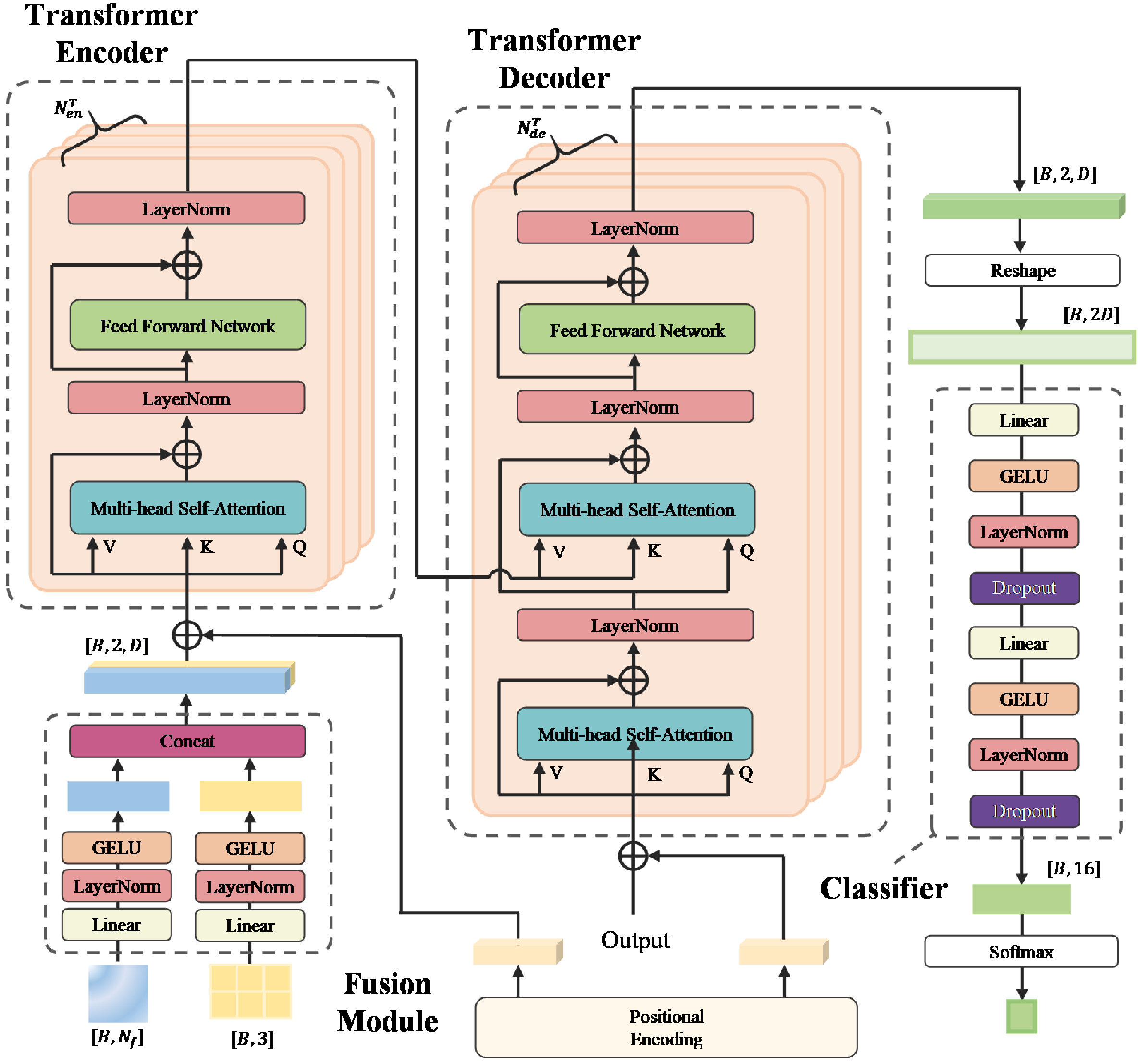}
	\caption{The network architecture of the proposed Fusion-TDC.}
	\label{fig:Transformer_class}
\end{figure}

The loss function $\mathcal{L}_{\mathrm{cla}}$ combines cross-entropy loss and label smooth loss, denoted as
\begin{equation}\label{eq:loss cla}
	\mathcal{L}_{\mathrm{cla}}\!=\!\frac{-1}{N_{\mathrm{tra}}}\!\!\sum_{k=1}^{N_{\mathrm{tra}}}\!\mu\!\overbrace{(\mathbf{Y}^{\mathrm{DO}}_{k})^T\!\log (\hat{\mathbf{Y}}^{\mathrm{DO}}_{k}\!)}^{\mathrm{cross-entropy\ loss}}\!\!+\!(\!1-\mu\!)\!\overbrace{(\tilde{\mathbf{Y}}^{\mathrm{DO}}_{k}\!)^T\!\!\log (\hat{\mathbf{Y}}^{\mathrm{DO}}_{k})}^{\mathrm{label\ smooth\ loss}},
\end{equation}
where 
\begin{equation}
	\left[\tilde{\mathbf{Y}}^{\mathrm{DO}}_{k}\right]_i=\begin{cases}
		1-\epsilon,& i \text{ is the true class}\\
		\frac{\epsilon}{C-1}, & \text{else}
	\end{cases},
\end{equation}
$C$ is the total number of classes, $\epsilon\in[0,1]$ is the smoothing factor, and $\mu$ is the weight coefficient. Properly adjusting $\epsilon$ and $\mu$ can mitigate overfitting.

In the fusion module, both inputs are processed by a linear layer, layer normalization, and GELU activation for feature projection. Then they are concatenated into $B\times2\times D$ tensor. The Transformer encoder and decoder are classical (detailed in \cite{Transformer}), where $N_{\mathrm{en}}^T$ and $N_{\mathrm{de}}^T$ in Fig. \ref{fig:Transformer_class} denote the numbers of encoder and decoder layers, respectively. Finally, the decoded tensor is reshaped into a matrix, which is subsequently fed through a classifier and Softmax function to generate the classification output.

\section{Simulation Results}\label{sec:simulation results}

In this section, we introduce the simulation setup and then evaluate the performance of LOA-Net in both localization and orientation awareness. To enhance the persuasiveness of the simulation, we chose a more challenging NLOS environment.

\subsection{Simulation Setup}

The simulation dataset is generated by the professional radio channel generation software QuaDRiGa \cite{quadriga}. After obtaining the perfect CSI through QuaDRiGa, additive white Gaussian noise when SNR$=$20 dB is superimposed to simulate noise and channel estimation errors. We configure the OFDM parameters in accordance with the 3GPP LTE \cite{3gpp36.211} specifications and generate the TBF for an indoor NLOS scenario, as defined in 3GPP 38.901 \cite{3gpp38.901}. Furthermore, the proposed method is applicable to any multi-subcarrier OFDM parameter configuration. The essential simulation parameters are summarized in Table \ref{tb:parameters}.

\newcolumntype{L}{>{\hspace*{-\tabcolsep}}l}
\newcolumntype{R}{c<{\hspace*{-\tabcolsep}}}
\definecolor{lightblue}{rgb}{0.93,0.95,1.0}
\definecolor{lightgreen}{rgb}{0.95,1.0,0.93}
\begin{table}[htbp]
	\captionsetup{font=footnotesize}
	\caption{Parameter Settings}\label{tb:parameters}
	\centering
	\ra{1.5}
	\scriptsize
	\begin{tabular}{LR}
		\toprule
		Parameter &  Value \\
		\rowcolor{lightblue}
		\midrule
		Number of BS antennas $A$ $(M_{\mathrm{r}}\times M_{\mathrm{c}})$ & 128 (16$\times$8)\\
		Carrier frequency &5.8 GHz\\
		\rowcolor{lightblue}
		Subcarrier interval $\Delta f$ &15 kHz  \\
		Subcarrier number $N_{\mathrm{c}}$&2048 \\
		\rowcolor{lightblue}
		CP length $N_{\mathrm{g}}$ & 144 \\
		Number of slot in each frame $N_{\mathrm{f}}$  & 8 \\
		\rowcolor{lightblue}
		Number of symbols in each slot $N_{\mathrm{s}}$ & 14\\
		Velocity of UTs & 5 km/h\\
		\bottomrule
	\end{tabular}
\end{table}

The BS with a height of 25 m is located in the middle of the cell, and the antenna plane is perpendicular to the ground. The heights of UTs are set to 1.5 m, 4.5 m, and 7.5 m to simulate floors; thus, all UTs are located in the far-field area of the BS. Considering the actual application scenarios such as hospitals, shopping malls, and museum lobbies, the sampling area is set at 40 m$\times$40 m. The sampling area is centered on the BS, and the TBFs are generated at an interval of 1 m in each floor. For the regression task, the training set comprises 5043 TBFs (41$\times$41 per floor) with randomized motion directions. The test set consists of 2100 TBFs (700 per floor) randomly distributed across the sampling area, following a 7:3 training-to-test ratio. In the multi-classification task, the training set includes TBFs from 16 distinct orientations, sampled at 5043 predefined positions. Given the substantial size of the training set (5043$\times$16 samples), the test set generates 19840 TBFs with random positions and orientations at a training-to-test ratio of less than 7:2. We train MaskDETR-Reg and Fusion-TDC separately, where the coordinate input of Fusion-TDC is replaced by the true label. The training regimen incorporates an adaptive learning rate schedule coupled with an early stopping policy, triggered by a patience of 100 epochs without improvement in the validation loss. Other hyperparameters are listed below.

\newcolumntype{L}{>{\hspace*{-\tabcolsep}}l}
\newcolumntype{R}{c<{\hspace*{-\tabcolsep}}}
\newcolumntype{C}{>{\centering\arraybackslash}p{2cm}}
\definecolor{lightblue}{rgb}{0.93,0.95,1.0}
\definecolor{lightgreen}{rgb}{0.95,1.0,0.93}
\begin{table}[htbp]
	\captionsetup{font=footnotesize}
	\caption{Training Hyperparameter Settings}\label{tb:hyperparameters}
	\centering
	\ra{1.5}
	\scriptsize
	\begin{tabular}{LCCR}
		\toprule
		Parameter & DE-MaskDETR & EC-MaskDETR & Fusion-TDC \\
		\rowcolor{lightblue}
		\midrule
		Learning Rate & $1\times 10^{-4}$ & $5\times 10^{-5}$ & $1\times 10^{-4}$\\
		Optimizer & Adam & Adam & AdamW \\
		\rowcolor{lightblue}
		Batch Size & 8 & 8 & 16\\
		Hidden Dim & 256 & 256 & 256 \\
		\rowcolor{lightblue}
		Number of Heads & 4 & 4 & 8 \\
		Encoder Layers & 4 & 6 & 4\\
		\rowcolor{lightblue}
		Decoder Layers & 2 & 2 & 2 \\
		Object Queries & 4 & 4 & N/A \\
		\bottomrule
	\end{tabular}
\end{table}

The TBFs of the training points are generated and saved using MATLAB R2023a, and the training and testing are processed using Pytorch 2.6. Our simulation is carried out on a workstation equipped with one 12th Gen Intel Core i7-12700 CPU and one NVIDIA GeForce RTX 3060 GPU. 

\subsection{Localization Performance}

To evaluate the performance of the proposed EC-MaskDETR and DE-MaskDETR, we adopt the CNN method from \cite{CWXY2021} and WKNN based on the Euclidean distance as the benchmarks. The localization error is computed as the $\ell_2$-norm between the estimated value and the ground-truth label. For the two control methods, 2D CNN and 3D CNN, only part of the convolution kernel sizes in the convolution refinement module are adjusted to fit the input size. The remaining hyperparameters are adjusted to ensure that each network converges well.

During the training phase, Fig. \ref{fig:loss} illustrates the loss curves of each network. The curves exhibit slight fluctuations and maintain a smooth descending trend throughout the training process. EC-MaskDETR achieves faster convergence than DE-MaskDETR due to the backbone. WKNN is a non-parametric learner and therefore has no loss convergence curve.

\begin{figure}[!h]
	\centering
	\includegraphics[width=3.7in]{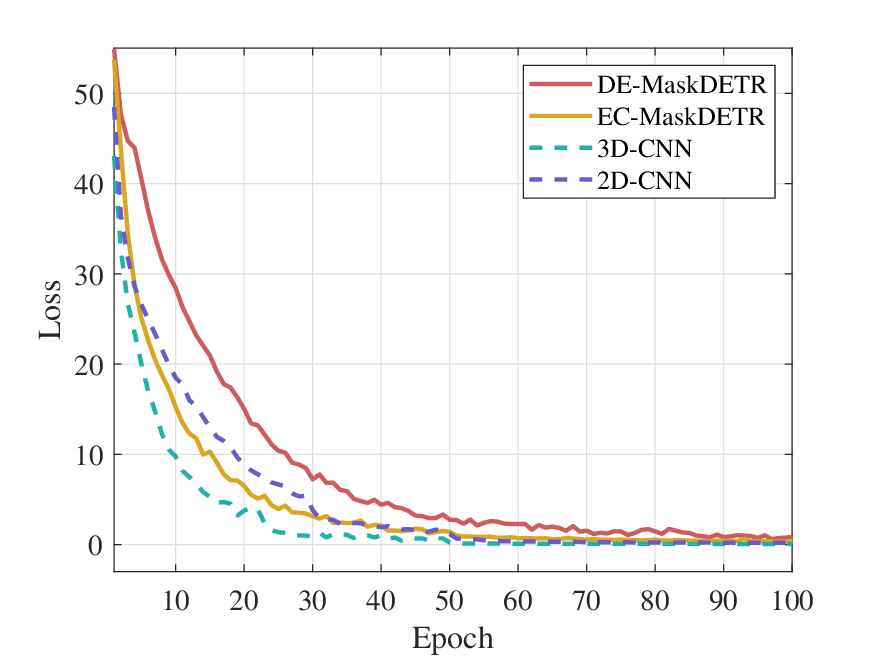}
	\caption{Loss function values for different models.}
	\label{fig:loss}
\end{figure}
\begin{figure}[!t]
	\centering
	\includegraphics[width=3.7in]{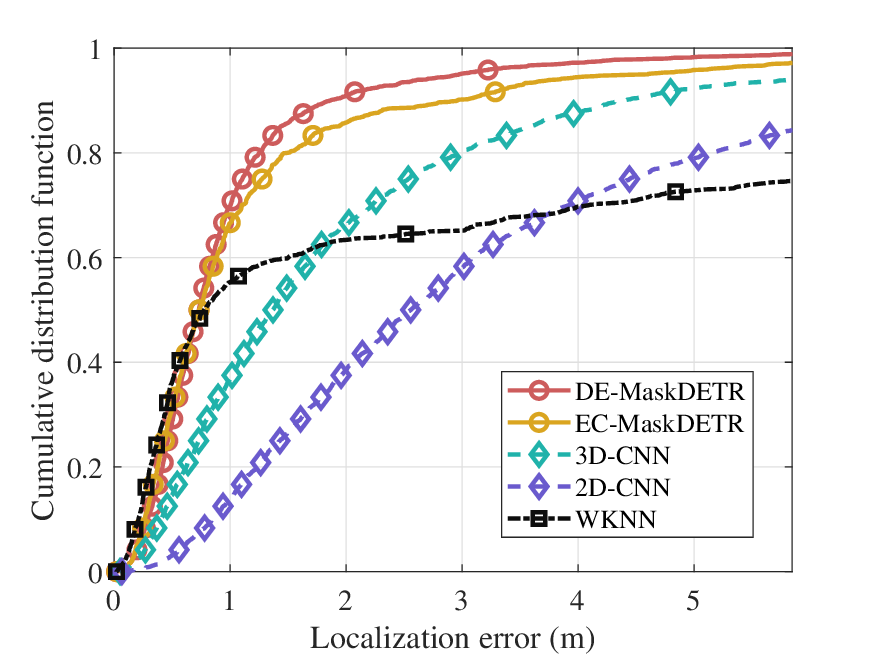}
	\caption{Cumulative distribution functions of localization errors for different localization methods.}
	\label{fig:err}
\end{figure}

Fig. \ref{fig:err} illustrates the cumulative distribution function (CDF) of the localization error of the test set. The DETR-based networks consistently outperform CNN-based networks in localization accuracy. Due to the difficulty of WKNN in establishing the high-dimensional and nonlinear relationship between TBF and position, its localization performance is unsatisfactory. From the results, it is evident that while some of the data estimates are relatively accurate, the majority exhibit significant errors, with the  indicating substantial overall deviation. DE-MaskDETR achieves the highest precision, with a mean error of 1.47 m and 90$\%$ of errors below 1.8 m. EC-MaskDETR exhibits lower performance, yielding an average error of 1.70 m. These results demonstrate that a feature-level backbone enables deeper feature extraction.

Table \ref{tb:Distance factor} presents the average localization accuracy of each method across different distance ranges. In the near-BS region (0-5 m), the scarcity of multipath signals leads to low fingerprint distinguishability, resulting in relatively higher localization errors. All fingerprint-based localization methods exhibit stable performance at greater distances.

\newcolumntype{L}{>{\hspace*{-\tabcolsep}}l}
\newcolumntype{R}{c<{\hspace*{-\tabcolsep}}}
\newcolumntype{C}{>{\centering\arraybackslash}p{0.7cm}}   
\newcolumntype{c}{>{\centering\arraybackslash}p{1.5cm}} 
\definecolor{lightblue}{rgb}{0.93,0.95,1.0}
\definecolor{lightgreen}{rgb}{0.95,1.0,0.93}
\begin{table}[htbp]
	\captionsetup{font=footnotesize}
	\caption{The Average Localization Error within Different Distance}\label{tb:Distance factor}
	\centering
	\ra{1.5}
	\scriptsize
	\begin{tabular}{LCCCcc}
		\toprule
		Distance (m) & WKNN & $\text{2D-CNN}$ & $\text{3D-CNN}$ & $\textbf{DE-MaskDETR}$ & $\textbf{EC-MaskDETR}$\\
		\rowcolor{lightblue}
		\midrule
		0-5 & 5.41 & 4.42 & 2.95& 1.92& 2.06\\
		5-10 & 4.58 & 3.56 & 2.31& 1.59& 1.6\\
		\rowcolor{lightblue}
		10-15 & 4.28 & 3.78 & 2.28& 1.25& 1.71\\
		15-20 & 4.49 & 3.71 & 2.43& 1.48& 1.71\\
		\bottomrule
	\end{tabular}
\end{table}

Table \ref{tb:SNR positioning} presents the average localization error of DE-MaskDETR under different SNR values. The results indicate that the error decreases with increasing SNR and gradually stabilizes at high SNR, demonstrating the robustness of the proposed method to noise.

\newcolumntype{L}{>{\hspace*{-\tabcolsep}}l}
\newcolumntype{R}{r<{\hspace*{-\tabcolsep}}}
\newcolumntype{C}{>{\centering\arraybackslash}p{0.3cm}}
\definecolor{lightblue}{rgb}{0.93,0.95,1.0}
\definecolor{lightgreen}{rgb}{0.95,1.0,0.93}
\begin{table}[htbp]
	\captionsetup{font=footnotesize}
	\caption{Average Localization Error under Different SNR Values}\label{tb:SNR positioning}
	\centering
	\ra{1.5}
	\scriptsize
	\begin{tabular}{LCCCCR}
		\toprule
		SNR (dB) & 0 & 5 & 10 & 15 & 20\\
		\rowcolor{lightblue}
		\midrule
		Average Localization Error (m) & 1.97 & 1.84 & 1.72& 1.63& 1.47\\
		\bottomrule
	\end{tabular}
\end{table}

\subsection{Orientation Awareness Performance}

Fig. \ref{fig:confusion mat} presents the confusion matrix derived from the test results at 5 km/h, where each cell displays the sample count (in parentheses) along with the corresponding recall rate (the proportion of true positives correctly identified). The diagonal elements of the confusion matrix exhibit significantly higher recall rates compared to the off-diagonal elements, demonstrating a strong classification accuracy. Specifically, the number of correctly classified samples in the test set is 14724, and the accuracy is 74.2$\%$. Notably, among the misclassified cases, those adjacent to the diagonal generally show higher recall rates, indicating that most of the misestimates are also close to the correct classification.

\begin{figure}[!t]
	\centering
	\includegraphics[width=3.6in]{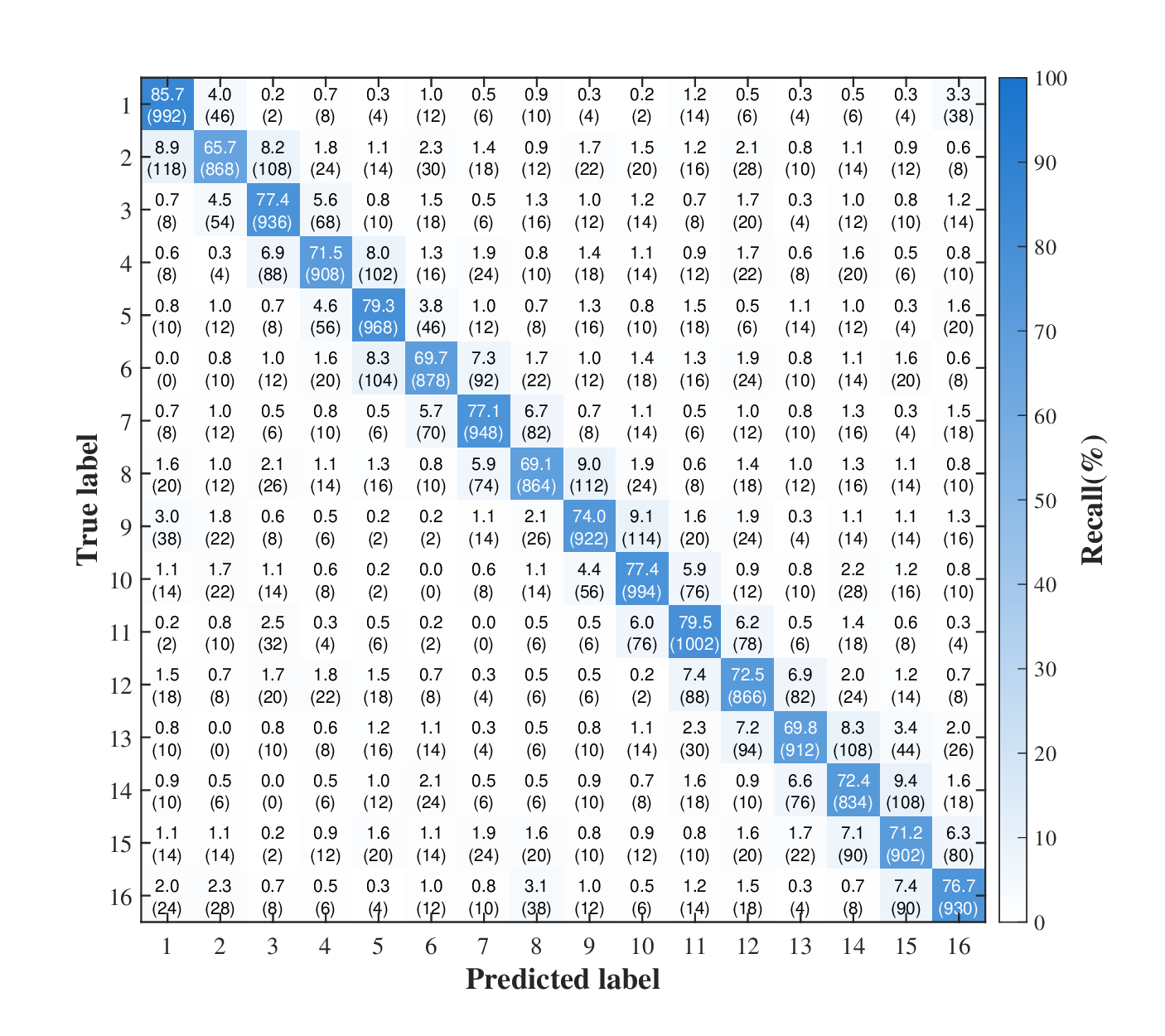}
	\caption{Confusion matrix for Fusion-TDC (5 km/h).}
	\label{fig:confusion mat}
\end{figure}

To investigate the influence of speed on orientation estimation, we have added simulations at 3 km/h, 8 km/h, and 30 km/h (corresponding to slow, fast walking, and uncrewed aircraft, respectively), and selected continuous CSI recording as the comparative method. To ensure fairness, the continuous CSI recording method uses two sets of CSI separated by one frame interval (the duration of a single TBF), with positions estimated by MaskDETR-Reg. The direction of the positional change is then taken as the orientation result. The simulation results are shown in Table \ref{tb:speed factor}. The proposed Fusion-TDC exhibits stable performance, with better results at higher speeds. Due to inherent localization errors and the short time interval, continuous CSI recording yields lower accuracy. This further demonstrates the advantage and feasibility of TBF for real-time orientation estimation.

\newcolumntype{L}{>{\hspace*{-\tabcolsep}}l}
\newcolumntype{R}{r<{\hspace*{-\tabcolsep}}}
\newcolumntype{C}{>{\centering\arraybackslash}p{1cm}}
\definecolor{lightblue}{rgb}{0.93,0.95,1.0}
\definecolor{lightgreen}{rgb}{0.95,1.0,0.93}
\begin{table}[htbp]
	\captionsetup{font=footnotesize}
	\caption{Orientation Estimation at Different Speeds}\label{tb:speed factor}
	\centering
	\ra{1.5}
	\scriptsize
	\begin{tabular}{LCCCR}
		\toprule
		Method & 3 km/h & 5 km/h & 8 km/h & 30 km/h\\
		\rowcolor{lightblue}
		\midrule
		\textbf{Fusion-TDC} & 72.9$\%$ & 74.2$\%$ & 74.7$\%$& 79.3$\%$\\
		Continuous CSI & 27.8$\%$ & 29.1$\%$ & 30.6$\%$& 49.7$\%$\\
		\bottomrule
	\end{tabular}
\end{table}

Table \ref{tb:SNR orientation} shows the orientation accuracy of the proposed Fusion-TDC method at 5 km/h under different SNR values. The results indicate that the accuracy improves slightly with increasing SNR and gradually stabilizes at high SNR levels, reaching 74.2$\%$ at 20 dB. This demonstrates that the proposed method is relatively robust to noise variations.

\newcolumntype{L}{>{\hspace*{-\tabcolsep}}l}
\newcolumntype{R}{r<{\hspace*{-\tabcolsep}}}
\newcolumntype{C}{>{\centering\arraybackslash}p{0.7cm}}
\definecolor{lightblue}{rgb}{0.93,0.95,1.0}
\definecolor{lightgreen}{rgb}{0.95,1.0,0.93}
\begin{table}[htbp]
	\captionsetup{font=footnotesize}
	\caption{Orientation Accuracy under Different SNR Values}\label{tb:SNR orientation}
	\centering
	\ra{1.5}
	\scriptsize
	\begin{tabular}{LCCCCR}
		\toprule
		SNR (dB) & 0 & 5 & 10 & 15 & 20\\
		\rowcolor{lightblue}
		\midrule
		Accuracy & 58.9$\%$ & 64.3$\%$ & 69.6$\%$& 72.5$\%$& 74.2$\%$\\
		\bottomrule
	\end{tabular}
\end{table}

\subsection{Complexity Analysis}

The computational and space complexities of the comparative models are summarized in Table \ref{tb:Complexity analysis}. 
\newcolumntype{L}{>{\hspace*{-\tabcolsep}}l}
\newcolumntype{R}{r<{\hspace*{-\tabcolsep}}}
\newcolumntype{C}{>{\centering\arraybackslash}p{2cm}}
\definecolor{lightblue}{rgb}{0.93,0.95,1.0}
\definecolor{lightgreen}{rgb}{0.95,1.0,0.93}
\begin{table}[htbp]
	\captionsetup{font=footnotesize}
	\caption{Computational and Space Complexity}\label{tb:Complexity analysis}
	\centering
	\ra{1.5}
	\scriptsize
	\begin{tabular}{LCR}
		\toprule
		Model & FlOPs & Param \\
		\rowcolor{lightblue}
		\midrule
		2D-CNN & 761.5M & 4.83M \\
		3D-CNN & 893.7M & 13.51M \\
		\rowcolor{lightblue}
		\textbf{DE-MaskDETR} & 296.3M & 3.39M \\
		\textbf{EC-MaskDETR} & 419.7M & 4.31M \\
		\rowcolor{lightblue}
		Fusion-TDC & 0.43M & 0.43M \\
		\bottomrule
	\end{tabular}
\end{table}

The DETR architecture effectively leverages the TBF and integrates with a mask mechanism, achieving high estimation accuracy even with a lightweight structure. In contrast, the CNN-based approach requires stacking numerous convolutional layers to process the entire TBF and applies uniform computations to both sparse and dense regions, leading to substantial redundant operations.

For the two backbone methods before the DETR framework, the convergence of DE-MaskDETR is relatively slow, but its performance is better. The faster convergence of EC-MaskDETR mainly results from its simpler backbone, which makes optimization easier during early training but limits feature representation capability, leading to slightly lower localization accuracy. Its higher complexity is primarily caused by the increased number of transformer encoder layers required to compensate for the relatively shallow backbone.

\section{Conclusion}\label{sec:conclusion}
In this paper, we utilized the TB domain channel power tensor as fingerprints and combined its structural features to propose a high-accuracy localization and orientation awareness network in a massive MIMO-OFDM system. Based on the TB domain channel model, we demonstrated that transforming SFTF to TBF not only reduced its size but also maintained the discriminability between fingerprints. Then we proposed LOA-Net composed of MaskDETR-Reg module and Fusion-TDC module. The MaskDETR-Reg module enhanced region-of-interest highlighting by fusing angle-delay domain information from TBF with masks, and estimated coordinates through DETR. The Fusion-TDC module combined Doppler domain information with estimated coordinates, and predicted motion direction using Transformer and a classifier. Simulation results demonstrated that the proposed method outperformed WKNN and CNN-based localization approaches while also achieving high accuracy in motion direction estimation. In the future, a meaningful topic would be to optimize fingerprint construction and network models around Doppler frequency, enabling simultaneous estimation of UT speed and motion direction.

\appendices
\section{Proof Of Theorem \ref{theorem:1}}\label{app:Theorem 1}
We demonstrate the discriminability of TBF across three dimensions of the TB domain. The subsequent three lemmas correspond to the investigation of angle, delay, and Doppler domains, respectively.

\begin{lemma}\label{lemma 1}
	With enough antennas, the angle transform matrix is able to acquire the DOA from the array response vector, i.e.
	\begin{align}
		\lim_{M_{\mathrm{c}}\to\infty}&\left(\frac{1}{\sqrt{M_{\mathrm{c}}}}(\mathbf{W}^{\phi}_{M_{\mathrm{c}}})^H\circ_{1}\mathbf{f}^{c}(\theta_{k,p})-\mathbf{\Lambda} _{M_{\mathrm{c}}}^{\breve{c}_{k,p}}\right)=\mathbf{0},\label{eq:angle column}\\
		\lim_{M_{\mathrm{r}}\to\infty}&\left(\frac{1}{\sqrt{M_{\mathrm{r}}}}(\mathbf{W}^{\phi}_{M_{\mathrm{r}}})^H\circ_{1}\mathbf{f}^{r}(\theta_{k,p},\varphi_{k,p})-\mathbf{\Lambda} _{M_{\mathrm{r}}}^{\breve{r}_{k,p}}\right)=\mathbf{0},\label{eq:angle row}
	\end{align}
\end{lemma}

\begin{IEEEproof}
	For the $i$-th element of $\frac{1}{\sqrt{M_{\mathrm{c}}}}(\mathbf{W}^{\phi}_{M_{\mathrm{c}}})^H\circ_{1}\mathbf{f}^{c}(\theta_{k,p})$,
	\begin{align}\label{eq:angle discriminability}
		&[\frac{1}{\sqrt{M_{\mathrm{c}}}}(\mathbf{W}^{\phi}_{M_{\mathrm{c}}})^H\circ_{1}\mathbf{f}^{c}(\theta_{k,p})]_{i}\nonumber\\
		=&\frac{1}{\sqrt{M_{\mathrm{c}}}}\sum_{j=0}^{M_{\mathrm{c}}-1}[(\mathbf{W}^{\phi}_{M_{\mathrm{c}}})^{H}]_{i,j}[\mathbf{f}^{c}(\theta_{k,p})]_{j}\nonumber\\
		=&\frac{1}{\sqrt{M_{\mathrm{c}}}}\sum_{j=0}^{M_{\mathrm{c}}-1}\frac{1}{\sqrt{M_{\mathrm{c}}}}e^{\bar{\imath}2\pi\frac{j(2i-M_{\mathrm{c}})}{2M_{\mathrm{c}}}}e^{-\bar{\imath}2\pi j\frac{d_{\mathrm{r}}}{\lambda_c}\cos\theta_{k,p}}\nonumber\\
		=&e^{-\bar{\imath}\pi\left(\frac{d_{\mathrm{r}}}{\lambda_c}\cos\theta_{k,p}-\frac{2i-M_{\mathrm{c}}}{2M_{\mathrm{c}}}\right)(M_{\mathrm{c}}-1)}\nonumber\\
		&\cdot \frac{\sin\Big(M_{\mathrm{c}}\big(\frac{d_{\mathrm{r}}}{\lambda_c}\cos\theta_{k,p}-\frac{2i-M_{\mathrm{c}}}{2M_{\mathrm{c}}}\big)\pi\Big)}{M_{\mathrm{c}}\sin\Big(\big(\frac{d_{\mathrm{r}}}{\lambda_c}\cos\theta_{k,p}-\frac{2i-M_{\mathrm{c}}}{2M_{\mathrm{c}}}\big)\pi\Big)}\nonumber\\
		=&e^{-\bar{\imath}\pi\left(\frac{d_{\mathrm{r}}}{\lambda_c}\cos\theta_{k,p}-\frac{2i-M_{\mathrm{c}}}{2M_{\mathrm{c}}}\right)(M_{\mathrm{c}}-1)}\nonumber\\
		&\cdot S_{M_{\mathrm{c}}}\Big(\frac{d_{\mathrm{r}}}{\lambda_c}\cos\theta_{k,p}-\frac{2i-M_{\mathrm{c}}}{2M_{\mathrm{c}}}\Big),
	\end{align}
	where 
	\begin{equation}\label{eq:frac sin}
		S_{L}(x)\triangleq\frac{\sin(L\pi x)}{L\sin(\pi x)},
	\end{equation}
	and $\lim_{L\to\infty}S_{L}(x)=\delta(x)$ for $-1<x<1$, where $\delta(x)$ is $1$ when $x=0$ and $0$ otherwise. Let
	\begin{equation}\label{eq:angle i}
		\breve{c}_{k,p}=\frac{M_{\mathrm{c}}d_{\mathrm{r}}}{\lambda_c}\cos\theta_{k,p}+\frac{M_{\mathrm{c}}}{2}.
	\end{equation}
	Notably, antenna spacing $d_{\mathrm{r}}$ satifies $d_{\mathrm{r}}=\frac{\lambda_c}{2}$ and $i<M_{\mathrm{c}}$, so $\left|\frac{d_{\mathrm{r}}}{\lambda_c}\cos\theta_{k,p}-\frac{2i-M_{\mathrm{c}}}{2M_{\mathrm{c}}}\right|<1$. We then have (\ref{eq:angle column}), and similarly (\ref{eq:angle row}).
\end{IEEEproof}

Similarly, we have the following Lemmas \ref{lemma 2} and \ref{lemma 3}.

\begin{lemma}\label{lemma 2}
	With a sufficient number of subcarriers, the delay transform matrix is able to acquire the TOA from the frequency domain steering vector, i.e.,
	\begin{equation}\label{eq:lemma delay}
		\lim_{N_{\mathrm{c}}\to\infty}\!\!\Big(\!\frac{1}{\sqrt{N_{\mathrm{c}}}}\!(\!\mathbf{W}^{\tau}_{N_{\mathrm{c}},N_{\mathrm{g}}}\!)\!^H\!\!\circ_{2}\!\mathbf{H}^{\mathrm{SF}}_{k}\!\!-\!\mathbf{f}^{\phi}\!(\!\theta_{k,p},\varphi_{k,p})\!\bullet\mathbf{\Lambda} _{N_{\mathrm{g}}}^{(\frac{\tau_{k,p}}{T_{\mathrm{s}}})}\Big)=\mathbf{0},
	\end{equation}
	where $\mathbf{H}^{\mathrm{SF}}_{k}=\mathbf{f}^{\phi}(\theta_{k,p},\varphi_{k,p})\bullet\mathbf{f}^{\tau}(\tau_{k,p})$ is the space-frequency domain phase factor matrix.
\end{lemma}

\begin{IEEEproof}
	Refer to (\ref{eq:angle discriminability}), $\frac{1}{\sqrt{N_{\mathrm{c}}}}(\mathbf{W}^{\tau}_{N_{\mathrm{c}},N_{\mathrm{g}}})^H\circ_{2}\mathbf{H}^{\mathrm{SF}}_{k}$ is  calculated column by column. For the $i$-th colnum,
	\begin{align}\label{eq:delay discriminability}
		&\left[\frac{1}{\sqrt{N_{\mathrm{c}}}}(\mathbf{W}^{\tau}_{N_{\mathrm{c}},N_{\mathrm{g}}})^H\circ_{2}\mathbf{H}^{\mathrm{SF}}_{k}\right]_{:,i}\nonumber\\
		=&\frac{1}{\sqrt{N_{\mathrm{c}}}}\sum_{j=0}^{N_{\mathrm{c}}-1}[(\mathbf{W}^{\tau}_{N_{\mathrm{c}},N_{\mathrm{g}}})^{H}]_{i,j}[\mathbf{H}^{\mathrm{SF}}_{k}]_{:,j}\nonumber\\
		=&\frac{1}{\sqrt{N_{\mathrm{c}}}}\sum_{j=0}^{N_{\mathrm{c}}-1}\frac{1}{\sqrt{N_{\mathrm{c}}}}e^{\bar{\imath}2\pi\frac{ji}{N_{\mathrm{c}}}}\mathbf{f}^{\phi}(\theta_{k,p},\varphi_{k,p})e^{-\bar{\imath} 2\pi j\tau_{k,p}\Delta f}\nonumber\\
		=&\mathbf{f}^{\phi}(\theta_{k,p},\varphi_{k,p})e^{-\bar{\imath}\pi\frac{(N_{\mathrm{c}}-1)}{N_{\mathrm{c}}}(\frac{\tau_{k,p}}{T_{\mathrm{s}}}-i)}S_{N_{\mathrm{c}}}\left(\frac{\frac{\tau_{k,p}}{T_{\mathrm{s}}}-i}{N_{\mathrm{c}}}\right).
	\end{align}
	Since $\tau_{k,p}<N_{\mathrm{g}}T_{\mathrm{s}}$ and $i\le N_{\mathrm{c}}-1$, we have
	\begin{equation}\label{eq:delay range}
		\left|\frac{\tau_{k,p}}{T_{\mathrm{s}}}-i\right|<N_{\mathrm{c}}.
	\end{equation}
	Hence, The only integer $\frac{1}{N_{\mathrm{c}}}(\frac{\tau_{k,p}}{T_{\mathrm{s}}}-i)$ is $0$. When $N_{\mathrm{c}}\to\infty$, $\displaystyle\lim_{N_{\mathrm{c}}\to\infty}S_{N_{\mathrm{c}}}\left(\frac{1}{N_{\mathrm{c}}}(\frac{\tau_{k,p}}{T_{\mathrm{s}}}-i)\right)\neq0$ if and only if $i=\frac{\tau_{k,p}}{T_{\mathrm{s}}}$. In this case, (\ref{eq:delay discriminability}) is represented as
	\begin{align}
		&\lim_{N_{\mathrm{c}}\to\infty}\Bigg\{\left[\frac{1}{\sqrt{N_{\mathrm{c}}}}(\mathbf{W}^{\tau}_{N_{\mathrm{c}},N_{\mathrm{g}}})^H\circ_{2}\mathbf{H}^{\mathrm{SF}}_{k}\right]_{:,i}\nonumber\\
		-&\mathbf{f}^{\phi}(\theta_{k,p},\varphi_{k,p})\delta( i-\frac{\tau_{k,p}}{T_{\mathrm{s}}})\Bigg\}=0.
	\end{align}
	Then, combining the values of each column, we get (\ref{eq:lemma delay}). This completes the proof.
\end{IEEEproof}

\begin{lemma}\label{lemma 3}
	When the number of OFDM symbols is sufficient, Doppler frequency can be acquired by the Doppler transform matrix from the time domain steering vector, i.e.,
	\begin{align}\label{eq:lemma doppler}
		&\lim_{N_{\mathrm{t}}\to\infty}\Big(\frac{1}{\sqrt{N_{\mathrm{t}}}}(\mathbf{W}^{\nu}_{N_{\mathrm{t}},N_{\mathrm{f}}})^H\circ_{3}\mathcal{H}^{\mathrm{SFT}}_k\nonumber\\
		-&\sum_{p=0}^{P-1}\!\beta_{k,p}\mathbf{f}^{\phi}(\theta_{k,p},\varphi_{k,p})\!\bullet\mathbf{f}^{\tau}(\tau_{k,p})\!\bullet\mathbf{\Lambda} _{N_{\mathrm{f}}}^{(N_{\mathrm{t}}\nu_{k,p}T_{\mathrm{sym}}+\frac{N_{\mathrm{f}}}{2})}\Big)=\mathbf{0}.
	\end{align}
\end{lemma}

\begin{IEEEproof}
	For $\frac{1}{\sqrt{N_{\mathrm{t}}}}(\mathbf{W}^{\nu}_{N_{\mathrm{t}},N_{\mathrm{f}}})^H\circ_{3}\mathcal{H}^{\mathrm{SFT}}_k$, we first calculate the $i$-th element of its third dimension,
	\begin{align}\label{eq:doppler discriminability}
		&\left[\frac{1}{\sqrt{N_{\mathrm{t}}}}(\mathbf{W}^{\nu}_{N_{\mathrm{t}},N_{\mathrm{f}}})^H\circ_{3}\mathcal{H}^{\mathrm{SFT}}_k\right]_{:,:,i}\nonumber\\
		=&\frac{1}{\sqrt{N_{\mathrm{t}}}}\sum_{j=0}^{N_{\mathrm{t}}-1}\left[(\mathbf{W}^{\nu}_{N_{\mathrm{t}},N_{\mathrm{f}}})^H\right]_{i,j}\left[\mathcal{H}^{\mathrm{SFT}}_k\right]_{:,:,j}\nonumber\\
		=&\frac{1}{N_{\mathrm{t}}}\sum_{j=0}^{N_{\mathrm{t}}-1}\sum_{p=0}^{P-1}\beta_{k,p}\mathbf{f}^{\phi}(\theta_{k,p},\varphi_{k,p})\bullet\mathbf{f}^{\tau}(\tau_{k,p})\nonumber\\
		&\cdot e^{\bar{\imath}2\pi(n_{\mathrm{T}}N_{\mathrm{s}} \nu_{k,p} T_{\mathrm{sym}}+\nu_{k,p} T_{\mathrm{sym}}j)}e^{-\bar{\imath}2\pi(n_{\mathrm{T}}+\frac{j}{N_{\mathrm{s}}})(\frac{2i-N_{\mathrm{f}}}{2N_{\mathrm{f}}})}\nonumber\\
		=&\sum_{p=0}^{P-1}\!\beta_{k,p}\mathbf{f}^{\phi}\!(\theta_{k,p},\varphi_{k,p})\!\bullet\!\mathbf{f}^{\tau}\!(\tau_{k,p})e^{\bar{\imath}2\pi n_{\mathrm{T}}(N_{\mathrm{s}} \nu_{k,p} T_{\mathrm{sym}}-\frac{2i-N_{\mathrm{f}}}{2N_{\mathrm{f}}})}\nonumber\\
		&\cdot \! e^{\bar{\imath}\pi(\nu_{k,p} T_{\mathrm{sym}}-\frac{2i-N_{\mathrm{f}}}{2N_{\mathrm{t}}})(\!N_{\mathrm{t}}-1\!)}S_{N_{\mathrm{t}}}(\nu_{k,p} T_{\mathrm{sym}}\!-\!\frac{2i-N_{\mathrm{f}}}{2N_{\mathrm{t}}}).
	\end{align}
	Notice that $\nu_{k,p}\in[-\frac{1}{2N_{\mathrm{s}}T_{\mathrm{sym}}},\frac{N_{\mathrm{f}}-2}{2N_{\mathrm{t}}T_{\mathrm{sym}}}]$ and $i\in[0,N_{\mathrm{f}}-1]$, we have
	\begin{equation}\label{eq:doppler range}
		\left|\nu_{k,p} T_{\mathrm{sym}}-\frac{2i-N_{\mathrm{f}}}{2N_{\mathrm{t}}}\right|<1.
	\end{equation}
	When $N_{\mathrm{t}}\to\infty$, $\displaystyle \lim_{N_{\mathrm{t}}\to\infty}S_{N_{\mathrm{t}}}(\nu_{k,p} T_{\mathrm{sym}}-\frac{2i-N_{\mathrm{f}}}{2N_{\mathrm{t}}})\neq0$ if and only if $i=N_{\mathrm{t}}\nu_{k,p}T_{\mathrm{sym}}+\frac{N_{\mathrm{f}}}{2}$. In this case, (\ref{eq:doppler discriminability}) is denoted as
	\begin{align}\label{eq:doppler layer}
		&\lim_{N_{\mathrm{t}}\to\infty}\Bigg\{\left[\frac{1}{\sqrt{N_{\mathrm{t}}}}(\mathbf{W}^{\nu}_{N_{\mathrm{t}},N_{\mathrm{f}}})^H\circ_{3}\mathcal{H}^{\mathrm{SFT}}_k\right]_{:,:,i}\nonumber\\
		-&\sum_{p=0}^{P-1}\!\beta_{k,p}\mathbf{f}^{\phi}\!(\!\theta_{k,p},\varphi_{k,p}\!)\!\bullet\!\mathbf{f}^{\tau}\!(\!\tau_{k,p}\!)\delta(i-N_{\mathrm{t}}\nu_{k,p}T_{\mathrm{sym}}\!-\!\frac{N_{\mathrm{f}}}{2})\!\!\Bigg\}\!\!=\!0.
	\end{align}
	According to (\ref{eq:doppler layer}), we can further derive (\ref{eq:lemma doppler}). This completes the proof.
\end{IEEEproof}

We next prove Theorem \ref{theorem:1}.
\begin{align}\label{eq:appendix A}
	&\mathcal{H}^{\mathrm{TB}}_k\nonumber\\
	\overset{(a)}{=}&\frac{1}{\sqrt{M_{\mathrm{c}}M_{\mathrm{r}}N_{\mathrm{c}}N_{\mathrm{t}}}}\big((\mathbf{W}^{\phi}_{M_{\mathrm{c}}})^H\otimes(\mathbf{W}^{\phi}_{M_{\mathrm{r}}})^H\big)\nonumber\\
	&\circ_{1}\Big((\mathbf{W}^{\tau}_{N_{\mathrm{c}},N_{\mathrm{g}}})^H\circ_{2}\big((\mathbf{W}^{\nu}_{N_{\mathrm{t}},N_{\mathrm{f}}})^H\circ_{3}\mathcal{H}^{\mathrm{SFT}}_k\big)\Big)\nonumber\\
	\overset{(d)}{=}&\frac{1}{\sqrt{M_{\mathrm{c}}M_{\mathrm{r}}N_{\mathrm{c}}}}\big((\mathbf{W}^{\phi}_{M_{\mathrm{c}}})^H \!\!\otimes\!(\mathbf{W}^{\phi}_{M_{\mathrm{r}}})^H\big)\!\circ_{1}\!\big((\mathbf{W}^{\tau}_{N_{\mathrm{c}},N_{\mathrm{g}}})^H\nonumber\\
	&\circ_{2}\!\sum_{p=0}^{P-1}\!\!\beta_{k,p}\!\mathbf{f}^{\phi}\!(\theta_{k,p},\varphi_{k,p})\!\bullet\!\mathbf{f}^{\tau}\!(\tau_{k,p})\bullet\mathbf{\Lambda} _{N_{\mathrm{f}}}^{(N_{\mathrm{t}}\nu_{k,p}T_{\mathrm{sym}}+\frac{N_{\mathrm{f}}}{2})}\big)\nonumber\\
	\overset{(c)}{=}&\frac{1}{\sqrt{M_{\mathrm{c}}M_{\mathrm{r}}}}\big((\mathbf{W}^{\phi}_{M_{\mathrm{c}}})^H \!\!\otimes\!(\mathbf{W}^{\phi}_{M_{\mathrm{r}}})^H\big)\circ_{1}\sum_{p=0}^{P-1}\beta_{k,p}\nonumber\\
	&\cdot\mathbf{f}^{\phi}\!(\theta_{k,p},\varphi_{k,p})\bullet\mathbf{\Lambda} _{N_{\mathrm{g}}}^{(\frac{\tau_{k,p}}{T_{\mathrm{s}}})}\bullet\mathbf{\Lambda} _{N_{\mathrm{f}}}^{(N_{\mathrm{t}}\nu_{k,p}T_{\mathrm{sym}}+\frac{N_{\mathrm{f}}}{2})}\nonumber\\
	=&\sum_{p=0}^{P-1}\beta_{k,p}\Big(\frac{1}{\sqrt{M_{\mathrm{c}}}}(\mathbf{W}^{\phi}_{M_{\mathrm{c}}})^H\circ_{1}\mathbf{f}^{c}(\theta_{k,p})\Big)\nonumber\\
	&\otimes\Big(\frac{1}{\sqrt{M_{\mathrm{r}}}}(\mathbf{W}^{\phi}_{M_{\mathrm{r}}})^H\circ_{1}\mathbf{f}^{r}(\theta_{k,p},\varphi_{k,p})\Big)\nonumber\\
	&\bullet\mathbf{\Lambda} _{N_{\mathrm{g}}}^{(\frac{\tau_{k,p}}{T_{\mathrm{s}}})}\bullet\mathbf{\Lambda} _{N_{\mathrm{f}}}^{(N_{\mathrm{t}}\nu_{k,p}T_{\mathrm{sym}}+\frac{N_{\mathrm{f}}}{2})}\nonumber\\
	\overset{(b)}{=}&\sum_{p=0}^{P-1}\!\beta_{k,p}\!\big(\!\mathbf{\Lambda} _{M_{\mathrm{c}}}^{\breve{c}_{k,p}}\!\!\otimes\!\mathbf{\Lambda} _{M_{\mathrm{r}}}^{\breve{r}_{k,p}}\!\big)\!\bullet\!\mathbf{\Lambda} _{N_{\mathrm{g}}}^{(\frac{\tau_{k,p}}{T_{\mathrm{s}}})}\!\!\bullet\!\mathbf{\Lambda} _{N_{\mathrm{f}}}^{(N_{\mathrm{t}}\nu_{k,p} T_{\mathrm{sym}}\!+\!\frac{N_{\mathrm{f}}}{2})},
\end{align}
where $(a)$ is from (\ref{eq:TB to SFT}). $(b)$, $(c)$ and $(d)$ are from Lemma \ref{lemma 1}, Lemma \ref{lemma 2} and Lemma \ref{lemma 3}. Since the different paths are independent and wide sense stationary, according to $\beta_{k,p}\sim\mathcal{CN}(0,\sigma^2_{k,p})$, we have
\begin{equation}\label{eq:beta*beta}
	\mathbb{E}\{\beta_{k,p}\beta_{k,p'}^{*}\}=\begin{cases}
		\sigma^2_{k,p},  &p=p'\\
		0, &p\neq p'
	\end{cases}.
\end{equation}
Then, substituting (\ref{eq:appendix A}) and (\ref{eq:beta*beta}) to (\ref{eq:TBF}), TBF can be denoted as follows,
\begin{align}
	&\lim_{M_{\mathrm{c}},M_{\mathrm{r}},N_{\mathrm{c}},N_{\mathrm{t}}\to\infty}\Big(\mathcal{F}_{k}\nonumber\\
	-&\sum_{p=0}^{P-1}\!\sigma^2_{k,p}\!\!\left(\!\mathbf{\Lambda} _{M_{\mathrm{c}}}^{\breve{c}_{k,p}}\!\!\otimes\!\mathbf{\Lambda} _{M_{\mathrm{r}}}^{\breve{r}_{k,p}}\!\right)\!\bullet\!\mathbf{\Lambda} _{N_{\mathrm{g}}}^{(\frac{\tau_{k,p}}{T_{\mathrm{s}}})}\!\!\bullet\!\mathbf{\Lambda} _{N_{\mathrm{f}}}^{(N_{\mathrm{t}}\nu_{k,p} T_{\mathrm{sym}}\!+\!\frac{N_{\mathrm{f}}}{2})}\Big)=\mathbf{0}.
\end{align}
This completes the proof.

\section{Proof Of Theorem \ref{theorem:2}}\label{app:Theorem 2}

To prove Theorem \ref{theorem:2}, we first define the second-order statistics of $\mathcal{H}^{\mathrm{TB}}_k$, as follows
\begin{equation}
	\mathcal{Y}_{k}\triangleq\mathbb{E}\left\{\mathcal{H}^{\mathrm{TB}}_k\bullet (\mathcal{H}^{\mathrm{TB}}_k)^{*}\right\}\in\mathbb{C}^{A\times N_{\mathrm{g}}\times N_{\mathrm{f}}\times A\times N_{\mathrm{g}}\times N_{\mathrm{f}}}.
\end{equation}
Since all elements of $\mathcal{H}^{\mathrm{TB}}_k$ are independent of each other, the non-pseudo-diagonal elements of $\mathcal{Y}_{k}$ are all zeros. According to Theorem \ref{theorem:1} and (\ref{eq:appendix A}), when $M_{\mathrm{c}}$, $M_{\mathrm{r}}$, $N_{\mathrm{c}}$ and $N_{\mathrm{t}}$ approach infinity, the pseudo-diagonal elements of $\mathcal{Y}_{k}$ are calculated as
\begin{align}\label{eq:TB sec-ord pseudo-diag}
	&\lim_{M_{\mathrm{c}},M_{\mathrm{r}},N_{\mathrm{c}},N_{\mathrm{t}}\to\infty}\Big(\left[\mathcal{Y}_{k}\right]_{i,j,l,i,j,l}-\sum_{p=0}^{P-1}\sigma^2_{k,p}\delta(j-\frac{\tau_{k,p}}{T_{\mathrm{s}}})\nonumber\\
	\cdot&\delta(i-\breve{c}_{k,p}M_{\mathrm{r}}-\breve{r}_{k,p})\delta(l-N_{\mathrm{t}}\nu_{k,p}T_{\mathrm{sym}}-\frac{N_{\mathrm{f}}}{2})\Big)=0.
\end{align}
Note that the pseudo-diagonal elements of $\mathcal{Y}_{k}$ are precisely equal to the elements of the TBF, i.e., $\left[\mathcal{Y}_{k}\right]_{i,j,l,i,j,l}=\left[\mathcal{F}_{k}\right]_{i,j,l}$. That implies,
\begin{align}\label{eq:Y to F}
	&\lim_{M_{\mathrm{c}},M_{\mathrm{r}},N_{\mathrm{c}},N_{\mathrm{t}}\to\infty}\left[\mathcal{Y}_{k}\right]_{i,j,l,i',j',l'}\nonumber\\
	=&\begin{cases}
		\lim_{M_{\mathrm{c}},M_{\mathrm{r}},N_{\mathrm{c}},N_{\mathrm{t}}\to\infty}\left[\mathcal{F}_{k}\right]_{i,j,l},&\!\! i,j,l=i',j',l'\\
		0,  &\!\! \text{else} 
	\end{cases}.
\end{align}

Given that $\mathcal{Y}_{k}$ and $\mathcal{X}_{k}$ have the same dimensions, we utilize $\mathcal{Y}_{k}$ to construct an intermediate variable, thereby simplifying the proof. We define $\dot{\mathcal{H}}^{\mathrm{TB}}_{k}\in\mathbb{C}^{A\times N_{\mathrm{c}}\times N_{\mathrm{f}}}$ as the delay domain extension of the TB domain channel response tensor,
\begin{align}
	\dot{\mathcal{H}}^{\mathrm{TB}}_{k}=&\frac{1}{\sqrt{M_{\mathrm{c}}M_{\mathrm{r}}N_{\mathrm{c}}N_{\mathrm{t}}}}(\mathbf{W}^{\nu}_{N_{\mathrm{t}},N_{\mathrm{f}}})^H\circ_{3}\bigg((\mathbf{W}^{\tau}_{N_{\mathrm{c}}})^H\nonumber\\
	&\circ_{2}\Big(\big((\mathbf{W}^{\phi}_{M_{\mathrm{c}}})^H\otimes(\mathbf{W}^{\phi}_{M_{\mathrm{r}}})^H\big)\circ_{1}\mathcal{H}^{\mathrm{SFT}}_k\Big)\bigg),
\end{align}
where $\mathbf{W}^{\tau}_{N_{\mathrm{c}}}\in\mathbb{C}^{N_{\mathrm{c}}\times N_{\mathrm{c}}}$ and it has the same definition for each element as (\ref{eq:delay transform matrices}). Subsequently, we define the delay-Doppler domain extension of $\mathcal{H}^{\mathrm{TB}}_{k}$ as follows,
\begin{align}
	\ddot{\mathcal{H}}^{\mathrm{TB}}_{k}=&\frac{1}{\sqrt{M_{\mathrm{c}}M_{\mathrm{r}}N_{\mathrm{c}}N_{\mathrm{t}}}}(\mathbf{W}^{\nu}_{N_{\mathrm{t}}})^H\circ_{3}\bigg((\mathbf{W}^{\tau}_{N_{\mathrm{c}}})^H\nonumber\\
	&\circ_{2}\Big(\big((\mathbf{W}^{\phi}_{M_{\mathrm{c}}})^H\otimes(\mathbf{W}^{\phi}_{M_{\mathrm{r}}})^H\big)\circ_{1}\mathcal{H}^{\mathrm{SFT}}_k\Big)\bigg),
\end{align}
where
\begin{equation}\label{eq:dopp extension trans mat}
	\left[\mathbf{W}^{\nu}_{N_{\mathrm{t}}}\right]_{i,j}\triangleq \frac{1}{\sqrt{N_{\mathrm{t}}}}e^{ \bar{\imath}2\pi(n_{\mathrm{T}}+\frac{i}{N_{\mathrm{s}}})\frac{(2j-N_{\mathrm{t}})}{2N_{\mathrm{t}}}},
\end{equation}
$\ddot{\mathcal{H}}^{\mathrm{TB}}_{k}\in\mathbb{C}^{A\times N_{\mathrm{c}}\times N_{\mathrm{t}}}$ and $\mathbf{W}^{\nu}_{N_{\mathrm{t}}}\in\mathbb{C}^{N_{\mathrm{t}}\times N_{\mathrm{t}}}$. The properties of $\dot{\mathcal{H}}^{\mathrm{TB}}_{k}$ and $\ddot{\mathcal{H}}^{\mathrm{TB}}_{k}$ are shown in the following lemma.
\begin{lemma}\label{lemma 4}
	For massive MIMO-OFDM systems, the delay extension of the TB domain channel response tensor tends to $\mathbf{I}_{N_{\mathrm{c}}\times N_{\mathrm{g}}}\circ_{2}\mathcal{H}^{\mathrm{TB}}_{k}$ when $N_{\mathrm{c}}$ approaches infinity, i.e., 
	\begin{equation}\label{eq:delay extension}
		\lim_{N_{\mathrm{c}}\to\infty}\Big(\dot{\mathcal{H}}^{\mathrm{TB}}_{k}-\mathbf{I}_{N_{\mathrm{c}}\times N_{\mathrm{g}}}\circ_{2}\mathcal{H}^{\mathrm{TB}}_{k}\Big)=\mathbf{0},
	\end{equation}
	where $\mathbf{I}_{N_{\mathrm{c}}\times N_{\mathrm{g}}}\in\mathbb{C}^{N_{\mathrm{c}}\times N_{\mathrm{g}}}$ and $\mathbf{I}_{N_{\mathrm{c}}\times N_{\mathrm{g}}}=\left[\mathbf{I}_{N_{\mathrm{c}}}\right]_{:,0:N_{\mathrm{g}}-1}$. When $N_{\mathrm{c}}\to\infty$ and $N_{\mathrm{t}}\to\infty$, the delay-Doppler domain extension trends to $\mathbf{I}_{N_{\mathrm{t}}\times N_{\mathrm{f}}}\circ_{3}\dot{\mathcal{H}}^{\mathrm{TB}}_{k}$, i.e.,
	\begin{equation}\label{eq:delay Dopp extension}
		\lim_{N_{\mathrm{c}},N_{\mathrm{t}}\to\infty}\Big(\ddot{\mathcal{H}}^{\mathrm{TB}}_{k}-\mathbf{I}_{N_{\mathrm{t}}\times N_{\mathrm{f}}}\circ_{3}\dot{\mathcal{H}}^{\mathrm{TB}}_{k}\Big)=\mathbf{0},
	\end{equation}
	where $\mathbf{I}_{N_{\mathrm{t}}\times N_{\mathrm{f}}}\in\mathbb{C}^{N_{\mathrm{t}}\times N_{\mathrm{f}}}$ and $\mathbf{I}_{N_{\mathrm{t}}\times N_{\mathrm{f}}}=\left[\mathbf{I}_{N_{\mathrm{t}}}\right]_{:,0:N_{\mathrm{s}}:N_{\mathrm{t}}-1}$.
\end{lemma}

\begin{IEEEproof}
	We first prove the properties of the delay domain extension of $\mathcal{H}^{\mathrm{TB}}_{k}$. With the operational property of $n$ mode product, we can derive
	\begin{align}
		&\mathbf{I}_{N_{\mathrm{c}}\times N_{\mathrm{g}}}\circ_{2}\mathcal{H}^{\mathrm{TB}}_{k}\nonumber\\
		=&\frac{1}{\sqrt{M_{\mathrm{c}}M_{\mathrm{r}}N_{\mathrm{c}}N_{\mathrm{t}}}}(\mathbf{W}^{\nu}_{N_{\mathrm{t}},N_{\mathrm{f}}})^H\circ_{3}\bigg(\mathbf{I}_{N_{\mathrm{c}}\times N_{\mathrm{g}}}(\mathbf{W}^{\tau}_{N_{\mathrm{c}},N_{\mathrm{g}}})^H\nonumber\\
		&\circ_{2}\Big(\big((\mathbf{W}^{\phi}_{M_{\mathrm{c}}})^H\otimes(\mathbf{W}^{\phi}_{M_{\mathrm{r}}})^H\big)\circ_{1}\mathcal{H}^{\mathrm{SFT}}_k\Big)\bigg).
	\end{align}
	Consequently, we have
	\begin{align}\label{eq:delay extension derive}
		&\lim_{N_{\mathrm{c}}\to\infty}\left(\dot{\mathcal{H}}^{\mathrm{TB}}_{k}-\mathbf{I}_{N_{\mathrm{c}}\times N_{\mathrm{g}}}\circ_{2}\mathcal{H}^{\mathrm{TB}}_{k}\right)\nonumber\\
		=&\lim_{N_{\mathrm{c}}\to\infty}\frac{1}{\sqrt{M_{\mathrm{c}}M_{\mathrm{r}}N_{\mathrm{c}}N_{\mathrm{t}}}}(\mathbf{W}^{\nu}_{N_{\mathrm{t}},N_{\mathrm{f}}})^H\nonumber\\
		&\circ_{3}\bigg(\check{\mathbf{W}}^{\tau}\circ_{2}\Big(\big((\mathbf{W}^{\phi}_{M_{\mathrm{c}}})^H\otimes(\mathbf{W}^{\phi}_{M_{\mathrm{r}}})^H\big)\circ_{1}\mathcal{H}^{\mathrm{SFT}}_k\Big)\bigg),
	\end{align}
	where
	\begin{align}
		\check{\mathbf{W}}^{\tau}=&(\mathbf{W}^{\tau}_{N_{\mathrm{c}}})^H-\mathbf{I}_{N_{\mathrm{c}}\times N_{\mathrm{g}}}(\mathbf{W}^{\tau}_{N_{\mathrm{c}},N_{\mathrm{g}}})^H\nonumber\\
		=&\left[\mathbf{0}_{N_{\mathrm{c}}\times N_{\mathrm{g}}},\check{\mathbf{W}}^{\tau}_{N_{\mathrm{c}},(N_{\mathrm{c}}-N_{\mathrm{g}})}\right]^H,
	\end{align}
	and $\check{\mathbf{W}}^{\tau}_{N_{\mathrm{c}},(N_{\mathrm{c}}-N_{\mathrm{g}})}\in\mathbb{C}^{N_{\mathrm{c}}\times (N_{\mathrm{c}}-N_{\mathrm{g}})}$ is the submatrix of the last $(N_{\mathrm{c}}-N_{\mathrm{g}})$ columns of $\mathbf{W}^{\tau}_{N_{\mathrm{c}}}$. Since the first $N_{\mathrm{g}}$ rows of $\check{\mathbf{W}}^{\tau}$ are zeros, the first $N_{\mathrm{g}}$ indices of the second dimension are all zeros after $2$ module product operation, i.e., $\displaystyle\lim_{N_{\mathrm{c}}\to\infty}\left[\dot{\mathcal{H}}^{\mathrm{TB}}_{k}-\mathbf{I}_{N_{\mathrm{c}}\times N_{\mathrm{g}}}\circ_{2}\mathcal{H}^{\mathrm{TB}}_{k}\right]_{:,0:N_{\mathrm{g}}-1,:}=\mathbf{0}$. 
	
	For the last $(N_{\mathrm{c}}-N_{\mathrm{g}})$ rows of $\check{\mathbf{W}}^{\tau}$,
	\begin{align}\label{eq:delay extension 2 model}
		&\left[\frac{1}{\sqrt{N_{\mathrm{c}}}}(\check{\mathbf{W}}^{\tau}_{N_{\mathrm{c}},(N_{\mathrm{c}}-N_{\mathrm{g}})})^H\circ_{2}\mathcal{H}^{\mathrm{SFT}}_{k}\right]_{:,i,:}\nonumber\\
		=&\frac{1}{N_{\mathrm{c}}}\!\!\!\sum_{j=0}^{N_{\mathrm{c}}-1}\!\!\sum_{p=0}^{P-1}\!\!\beta_{k,p}\mathbf{f}^{\phi}\!(\theta_{k,p},\!\varphi_{k,p})e^{\bar{\imath}2\pi\frac{j(i+N_{\mathrm{g}})}{N_{\mathrm{c}}}}\!\! e^{-\bar{\imath} 2\pi j\tau_{k,p}\Delta f}\!\!\bullet\!\mathbf{f}^{\nu}(\nu_{k,p})\nonumber\\
		=&\sum_{p=0}^{P-1}\beta_{k,p}\mathbf{f}^{\phi}(\theta_{k,p},\varphi_{k,p})e^{-\bar{\imath}\pi\frac{(N_{\mathrm{c}}-1)}{N_{\mathrm{c}}}(\frac{\tau_{k,p}}{T_{\mathrm{s}}}-i-N_{\mathrm{g}})}\nonumber\\
		&\cdot S_{N_{\mathrm{c}}}\left(\frac{(\frac{\tau_{k,p}}{T_{\mathrm{s}}}-i-N_{\mathrm{g}})}{N_{\mathrm{c}}}\right)\bullet\mathbf{f}^{\nu}(\nu_{k,p}).
	\end{align}
	The remaining proof of (\ref{eq:delay extension}) is similar that of the proof of Lemma \ref{lemma 1}, and (\ref{eq:delay Dopp extension}) can be similarly proved as well.
	
	Next, we will adopt a similar approach to demonstrate the properties of the delay-Doppler domain extension, i.e.,
	\begin{align}\label{eq:delay dopp extension derive}
		&\lim_{N_{\mathrm{c}},N_{\mathrm{t}}\to\infty}\left(\ddot{\mathcal{H}}^{\mathrm{TB}}_{k}-\mathbf{I}_{N_{\mathrm{t}}\times N_{\mathrm{f}}}\circ_{3}\dot{\mathcal{H}}^{\mathrm{TB}}_{k}\right)\nonumber\\
		=&\lim_{N_{\mathrm{c}},N_{\mathrm{t}}\to\infty}\frac{1}{\sqrt{M_{\mathrm{c}}M_{\mathrm{r}}N_{\mathrm{c}}N_{\mathrm{t}}}}\check{\mathbf{W}}^{\nu}\circ_{3}\bigg((\mathbf{W}^{\tau}_{N_{\mathrm{c}}})^H\nonumber\\
		&\circ_{2}\Big(\big((\mathbf{W}^{\phi}_{M_{\mathrm{c}}})^H\otimes(\mathbf{W}^{\phi}_{M_{\mathrm{r}}})^H\big)\circ_{1}\mathcal{H}^{\mathrm{SFT}}_k\Big)\bigg),
	\end{align}
	where
	\begin{align}
		\check{\mathbf{W}}^{\nu}=&(\mathbf{W}^{\nu}_{N_{\mathrm{t}}})^H-\mathbf{I}_{N_{\mathrm{t}}\times N_{\mathrm{f}}}(\mathbf{W}^{\nu}_{N_{\mathrm{t}},N_{\mathrm{f}}})^H\nonumber\\
		=&\left[0,\check{\mathbf{W}}^{\nu}_{N_{\mathrm{t}},1:N_{\mathrm{s}}-1},\cdots,0,\check{\mathbf{W}}^{\nu}_{N_{\mathrm{t}},(N_{\mathrm{f}}-1)N_{\mathrm{s}}+1:N_{\mathrm{f}}N_{\mathrm{s}}-1}\right]^H.
	\end{align}
	and $\check{\mathbf{W}}^{\nu}_{N_{\mathrm{t}},iN_{\mathrm{s}}+1:(i+1)N_{\mathrm{s}}-1}\in\mathbb{C}^{N_{\mathrm{t}}\times (N_{\mathrm{s}}-2)}(i=0,\cdots,N_{\mathrm{f}}-1)$ is the $(iN_{\mathrm{s}}+1)\text{th}$ through $\big((i+1)N_{\mathrm{s}}-1\big)\text{th}$ column element of $\mathbf{W}^{\nu}_{N_{\mathrm{t}}}$. For rows of zero in $\check{\mathbf{W}}^{\nu}$, we have $\displaystyle\left[\lim_{N_{\mathrm{c}},N_{\mathrm{t}}\to\infty}\left(\ddot{\mathcal{H}}^{\mathrm{TB}}_{k}-\mathbf{I}_{N_{\mathrm{t}}\times N_{\mathrm{f}}}\circ_{3}\dot{\mathcal{H}}^{\mathrm{TB}}_{k}\right)\right]_{:,:,0:N_{\mathrm{s}}:N_{\mathrm{t}}-1}=\mathbf{0}$. For the non-zero $\check{\mathbf{W}}^{\nu}_{N_{\mathrm{t}},iN_{\mathrm{s}}+1:(i+1)N_{\mathrm{s}}-1}$, we derive it as follows,
	\begin{align}\label{eq:delay dopp extension 3 model}
		&\left[\frac{1}{\sqrt{N_{\mathrm{t}}}}(\mathbf{W}^{\nu}_{N_{\mathrm{t}}})^H\circ_{3}\mathcal{H}^{\mathrm{SFT}}_k\right]_{:,:,j}\nonumber\\
		=&\sum_{p=0}^{P-1}\!\beta_{k,p}\mathbf{f}^{\phi}\!(\theta_{k,p},\varphi_{k,p})\!\bullet\!\mathbf{f}^{\tau}\!(\tau_{k,p})e^{\bar{\imath}2\pi n_{\mathrm{T}}(N_{\mathrm{s}} \nu_{k,p} T_{\mathrm{sym}}-\frac{2j-N_{\mathrm{t}}}{2N_{\mathrm{t}}})}\nonumber\\
		&\cdot \! e^{\bar{\imath}\pi(\nu_{k,p} T_{\mathrm{sym}}-\frac{2j-N_{\mathrm{t}}}{2N_{\mathrm{t}}N_{\mathrm{s}}})(\!N_{\mathrm{t}}-1\!)}S_{N_{\mathrm{t}}}(\nu_{k,p} T_{\mathrm{sym}}\!-\!\frac{2j-N_{\mathrm{t}}}{2N_{\mathrm{t}}N_{\mathrm{s}}}).
	\end{align}
	Then, (\ref{eq:delay Dopp extension}) can be similarly proved as that in Lemma \ref{lemma 1} as well.
\end{IEEEproof}

We define the second-order statistics of $\ddot{\mathcal{H}}^{\mathrm{TB}}_{k}$, as follows
\begin{align}\label{eq:second statis ddot y}
	\ddot{\mathcal{Y}}_{k}\triangleq&\mathbb{E}\left\{\ddot{\mathcal{H}}^{\mathrm{TB}}_{k}\bullet (\ddot{\mathcal{H}}^{\mathrm{TB}}_{k})^{*}\right\}\nonumber\\
	=&(\mathbf{W}^{\nu}_{N_{\mathrm{t}}})^T \circ_{6}\Bigg((\mathbf{W}^{\tau}_{N_{\mathrm{c}}})^T\circ_{5}\bigg(((\mathbf{W}^{\phi}_{M_{\mathrm{c}}})^T\otimes(\mathbf{W}^{\phi}_{M_{\mathrm{r}}})^T)\nonumber\\
	&\circ_{4}\Big((\mathbf{W}^{\nu}_{N_{\mathrm{t}}})^H\circ_{3}\big((\mathbf{W}^{\tau}_{N_{\mathrm{c}}})^H\circ_{2}\big(((\mathbf{W}^{\phi}_{M_{\mathrm{c}}})^H\otimes(\mathbf{W}^{\phi}_{M_{\mathrm{r}}})^H)\nonumber\\
	&\circ_{1}(\mathcal{H}^{\mathrm{SFT}}_k\bullet(\mathcal{H}^{\mathrm{SFT}}_k)^{*})\big)\big)\Big)\bigg)\Bigg)\cdot \frac{1}{M_{\mathrm{c}}M_{\mathrm{r}}N_{\mathrm{c}}N_{\mathrm{t}}}\nonumber\\
	=&\frac{1}{M_{\mathrm{c}}M_{\mathrm{r}}N_{\mathrm{c}}N_{\mathrm{t}}}\ddot{\mathcal{W}}\diamond_{6}\mathcal{X}_{k},
\end{align}
where $\ddot{\mathcal{W}}\in\mathbb{C}^{A\times N_{\mathrm{c}}\times N_{\mathrm{t}}\times A\times N_{\mathrm{c}}\times N_{\mathrm{t}}\times A\times N_{\mathrm{c}}\times N_{\mathrm{t}}\times A\times N_{\mathrm{c}}\times N_{\mathrm{t}}}$ and
\begin{align}
	&\left[\ddot{\mathcal{W}}\right]_{i_1,i_2,i_3,i_4,i_5,i_6,j_1,j_2,j_3,j_4,j_5,j_6}\nonumber\\
	\triangleq&\left[(\mathbf{W}^{\phi}_{M_{\mathrm{c}}})^H \!\otimes(\mathbf{W}^{\phi}_{M_{\mathrm{r}}})^H\right]_{i_1,j_1}\!\left[(\mathbf{W}^{\tau}_{N_{\mathrm{c}}})^H\right]_{i_2,j_2}\left[(\mathbf{W}^{\nu}_{N_{\mathrm{t}}})^H\right]_{i_3,j_3}\nonumber\\
	&\cdot\left[(\mathbf{W}^{\phi}_{M_{\mathrm{c}}})^T \!\otimes(\mathbf{W}^{\phi}_{M_{\mathrm{r}}})^T\right]_{i_4,j_4}\!\left[(\mathbf{W}^{\tau}_{N_{\mathrm{c}}})^T\right]_{i_5,j_5}\left[(\mathbf{W}^{\nu}_{N_{\mathrm{t}}})^T\right]_{i_6,j_6}.
\end{align}
Equation (\ref{eq:second statis ddot y}) shows the transformation relationship between $\ddot{\mathcal{Y}}_{k}$ and $\mathcal{X}_{k}$. Then, we calculate the traces of the tensors and convert them into the element-wise summation operations according to (\ref{eq:tensor trace}), denoted as
\begin{equation}\label{eq:trace Y extension}
	\text{Tr}\{\ddot{\mathcal{Y}}_{k}\bullet\ddot{\mathcal{Y}}_{k'}^{*}\}=\frac{\text{Sum}\{(\ddot{\mathcal{W}}\diamond_{6}\mathcal{X}_{k})\odot(\ddot{\mathcal{W}}\diamond_{6}\mathcal{X}_{k'})^{*}\}}{(M_{\mathrm{c}}M_{\mathrm{r}}N_{\mathrm{c}}N_{\mathrm{t}})^{2}}.
\end{equation}
where $\text{Sum}\{\cdot\}$ denotes the summation of all the elements of a tensor. The following lemma describes a property of the combination of $\text{Sum}\{\cdot\}$ and $n$ module product, which is beneficial in the subsequent derivation process.

\begin{lemma}\label{lemma 5}
	Regarding the tensors $\mathcal{T}_{1}\in\mathbb{C}^{I_0\times \cdots\times I_N\times\cdots\times I_M}$ and $\mathcal{T}_{2}\in\mathbb{C}^{I_0\times \cdots\times I_N\times\cdots\times I_M}$, they exhibit the following operational relationship with the $N+1$ module product of an unitary matrix $\mathbf{O}\in\mathbb{C}^{I_N\times I_N}$, where $N+1\le M$:
	\begin{equation}
		\text{Sum}\{(\mathbf{O}\circ_{N+1}\mathcal{T}_{1})\odot(\mathbf{O}\circ_{N+1}\mathcal{T}_{2})^{*}\}=\text{Sum}\{\mathcal{T}_{1}\odot\mathcal{T}_{2}^{*}\}.
	\end{equation}
\end{lemma}

\begin{IEEEproof}
	This lemma can be easily proved by rearranging the order of the summation of each element. Taking the $N+1$ module product as an example, we have
	\begin{align}\label{eq:sum lemma}
		&\text{Sum}\{(\mathbf{O}\circ_{N+1}\mathcal{T}_{1})\odot(\mathbf{O}\circ_{N+1}\mathcal{T}_{2})^{*}\}\nonumber\\
		=&\sum_{i=0}^{I_N-1}\sum_{j=0}^{I_N-1}\Big(\text{Sum}\{\left[\mathbf{O}\right]_{:,i}\odot\left[\mathbf{O}^{*}\right]_{:,j}\}\nonumber\\
		&\cdot\text{Sum}\{\underbrace{\left[\mathcal{T}_{1}\right]_{:,\cdots,i,\cdots,:}\odot\left[\mathcal{T}_{2}^{*}\right]_{:,\cdots,j,\cdots,:}}_{i,j \text{are in the}(N+1) \text{th dimension}}\}\Big).
	\end{align}
	Since $\mathbf{O}$ is a square matrix, the summation operation can be simplified to matrix multiplication,
	\begin{equation}
		\text{Sum}\{\left[\mathbf{O}\right]_{:,i}\odot\left[\mathbf{O}^{*}\right]_{:,j}\}=\left[\mathbf{O}^{H}\mathbf{O}\right]_{j,i}.
	\end{equation}
	Leveraging the property of the unitary matrix, $\mathbf{O}^{H}\mathbf{O}=\mathbf{I}_{I_N}$, $\text{Sum}\{\left[\mathbf{O}\right]_{:,i}\odot\left[\mathbf{O}^{*}\right]_{:,j}\}\neq 0$ only if $i=j$. Consequently, (\ref{eq:sum lemma}) is further simplified as follows,
	\begin{align}
		&\text{Sum}\{(\mathbf{O}\circ_{N+1}\mathcal{T}_{1})\odot(\mathbf{O}\circ_{N+1}\mathcal{T}_{2})^{*}\}\nonumber\\
		=&\sum_{i=0}^{I_N-1}\text{Sum}\{\left[\mathcal{T}_{1}\right]_{:,\cdots,i,\cdots,:}\odot\left[\mathcal{T}_{2}^{*}\right]_{:,\cdots,i,\cdots,:}\}\nonumber\\
		=&\text{Sum}\{\mathcal{T}_{1}\odot\mathcal{T}_{2}^{*}\}.
	\end{align}
	This completes the proof.
\end{IEEEproof}

The transform matrices corresponding to the three beam domains are all generated by the unitary discrete Fourier transform (DFT) matrix \cite{JTXG2024}, hence $\mathbf{W}^{\nu}_{N_{\mathrm{t}}}$, $\mathbf{W}^{\tau}_{N_{\mathrm{c}}}$ and $(\mathbf{W}^{\phi}_{M_{\mathrm{c}}}\otimes\mathbf{W}^{\phi}_{M_{\mathrm{r}}})$ are unitary matrices. By substituting (\ref{eq:second statis ddot y}) into (\ref{eq:trace Y extension}) and applying Lemma \ref{lemma 5} to each dimension, we have
\begin{equation}\label{eq:trace Y exten to X}
	\text{Tr}\{\ddot{\mathcal{Y}}_{k}\bullet\ddot{\mathcal{Y}}_{k'}^{*}\}=\frac{\text{Sum}\{\mathcal{X}_{k}\odot\mathcal{X}_{k'}^{*}\}}{(M_{\mathrm{c}}M_{\mathrm{r}}N_{\mathrm{c}}N_{\mathrm{t}})^{2}}=\frac{\text{Tr}\{\mathcal{X}_{k}\bullet\mathcal{X}_{k'}^{*}\}}{(M_{\mathrm{c}}M_{\mathrm{r}}N_{\mathrm{c}}N_{\mathrm{t}})^{2}}.
\end{equation}

According to Lemma \ref{lemma 4},
\begin{equation}
	\lim_{N_{\mathrm{c}},N_{\mathrm{t}}\to\infty}\bigg(\ddot{\mathcal{H}}^{\mathrm{TB}}_{k}-\mathbf{I}_{N_{\mathrm{t}}\times N_{\mathrm{f}}}\circ_{3}\left(\mathbf{I}_{N_{\mathrm{c}}\times N_{\mathrm{g}}}\circ_{2}\mathcal{H}^{\mathrm{TB}}_{k}\right)\bigg)=\mathbf{0}.
\end{equation}
\begin{align}\label{eq:Y exten to Y}
	\lim_{N_{\mathrm{c}},N_{\mathrm{t}}\to\infty}\Bigg(\ddot{\mathcal{Y}}_{k}-&\mathbf{I}_{N_{\mathrm{t}}\times N_{\mathrm{f}}}^{*}\circ_{6}\bigg(\mathbf{I}_{N_{\mathrm{c}}\times N_{\mathrm{g}}}^{*}\circ_{5}\Big(\mathbf{I}_{N_{\mathrm{t}}\times N_{\mathrm{f}}}\nonumber\\
	&\circ_{3}\big(\mathbf{I}_{N_{\mathrm{c}}\times N_{\mathrm{g}}}\circ_{2}\mathcal{Y}_{k}\big)\Big)\bigg)\Bigg)=\mathbf{0}.
\end{align}
When $M_{\mathrm{c}}$, $M_{\mathrm{r}}$, $N_{\mathrm{c}}$ and $N_{\mathrm{t}}$ are sufficiently large, $\text{Tr}\{\ddot{\mathcal{Y}}_{k}\bullet\ddot{\mathcal{Y}}_{k'}^{*}\}$ can also be calculated as
\begin{align}\label{eq:trace Y exten to F}
	\text{Tr}\{\ddot{\mathcal{Y}}_{k}\bullet\ddot{\mathcal{Y}}_{k'}^{*}\}=&\text{Sum}\{\ddot{\mathcal{Y}}_{k}\odot\ddot{\mathcal{Y}}_{k'}^{*}\}\overset{(\text{a})}{\approx}\text{Sum}\{\mathcal{Y}_{k}\odot\mathcal{Y}_{k'}^{*}\}\nonumber\\
	\overset{(\text{b})}{\approx}&\text{Sum}\{\mathcal{F}_{k}\odot\mathcal{F}_{k'}^{*}\}=\text{Tr}\{\mathcal{F}_{k}\bullet\mathcal{F}_{k'}^{*}\},
\end{align}
where (a) follows from (\ref{eq:Y exten to Y}) and Lemma \ref{lemma 5}, (b) follows from (\ref{eq:Y to F}). Referring to (\ref{eq:trace Y exten to X}) and (\ref{eq:trace Y exten to F}), we have $\text{Tr}\{\mathcal{F}_{k}\bullet\mathcal{F}_{k'}^{*}\}\approx\frac{\text{Tr}\{\mathcal{X}_{k}\bullet\mathcal{X}_{k'}^{*}\}}{(M_{\mathrm{c}}M_{\mathrm{r}}N_{\mathrm{c}}N_{\mathrm{t}})^{2}}$. For Frobenius norm, we have
\begin{align}
	\left \| \mathcal{X}_{k} \right \|_\mathrm{F}=\sqrt{\text{Sum}\{\mathcal{X}_{k}\odot\mathcal{X}_{k}^{*}\}}\approx M_{\mathrm{c}}M_{\mathrm{r}}N_{\mathrm{c}}N_{\mathrm{t}}\left \| \mathcal{F}_{k} \right \|_\mathrm{F}.
\end{align}
Finally, combined with the definition of collinearity in (\ref{eq:coll define}), (\ref{eq:coll equal}) can be proved. This completes the proof.

\bibliographystyle{IEEEtran}
\bibliography{EE_AI}

\end{document}